\documentclass{article}
\usepackage[letterpaper, portrait, margin=0.75in]{geometry}

\usepackage{lipsum}
\usepackage{amsfonts}
\usepackage{graphicx}
\usepackage{epstopdf}
\usepackage{algorithmic}
\ifpdf
  \DeclareGraphicsExtensions{.eps,.pdf,.png,.jpg}
\else
  \DeclareGraphicsExtensions{.eps}
\fi
\usepackage{url}

\usepackage{enumitem}
\setlist[enumerate]{leftmargin=.5in}
\setlist[itemize]{leftmargin=.5in}

\usepackage{amsthm}

\newtheorem{theorem}{Theorem}

\newtheorem{remark}{Remark}

\newtheorem{example}{Example}
\newtheorem{lemma}{Lemma}

\newtheorem{condition}{Condition}

\newtheorem{definition}{Definition}

\usepackage{amsmath}
\usepackage{mathtools}
\usepackage{mathrsfs}
\usepackage{chemfig}
\usepackage{bm}
\usepackage{subfig}
\usepackage{array,etoolbox}

\numberwithin{theorem}{section} %
\numberwithin{lemma}{section} %
\numberwithin{definition}{section} %
\numberwithin{remark}{section} %

\newcounter{rownumber}[figure] 

\preto\tabular{\setcounter{magicrownumbers}{0}}
\newcounter{magicrownumbers}

\DeclarePairedDelimiter\norm{\lVert}{\rVert}%

\DeclareMathOperator{\vspan}{span}

\DeclareMathOperator{\rank}{rank}

\renewcommand\vec{\bm}

\newcommand{\retainlabel}[1]{\label{#1}\sbox0{\ref{#1}}}

\title{Identifiability of Chemical Reaction Networks with Intrinsic and Extrinsic Noise from Stationary Distributions\thanks{
This work was funded in part by the U.S. National Science Foundation under Grant CMMI 1727189 and the U.S. AFOSR MURI under grant FA9550-22-1-0316.}}

\author{Theodore W. Grunberg\thanks{Department of Electrical Engineering and Computer Science, MIT, Cambridge, MA 02139 USA 
  (\texttt{grunberg@mit.edu}).}
\and Domitilla Del Vecchio\thanks{Department of Mechanical Engineering, MIT, Cambridge, 
MA 02139 USA 
  (\texttt{ddv@mit.edu}).}
}

\usepackage{amsopn}
\DeclareMathOperator{\diag}{diag}

\begin{document}

\maketitle

\begin{abstract}
Many biological systems can be modeled as a chemical reaction network with unknown parameters. Data available to identify these parameters are often in the form of a stationary distribution, such as that obtained from measurements of a cell population. In this work, we introduce a framework for analyzing the identifiability of the reaction rate coefficients of chemical reaction networks from stationary distribution data. Working with the linear noise approximation, which is a diffusive approximation to the chemical master equation, we give a computational procedure to certify global identifiability based on Hilbert's Nullstellensatz. We present a variety of examples that show the applicability of our method to chemical reaction networks of interest in systems and synthetic biology, including discrimination between possible molecular mechanisms for the interaction between biochemical species. 
\end{abstract}


\section{Introduction}\label{sec:introduction}
System identification is concerned with going from a model class for a system to a particular model in that class based on experimental data. The basic property that guarantees that this is possible with sufficient data is \emph{structural identifiability}~\cite{Bellman:1970ue}. One practical use of identifiability analysis is to determine whether a particular experimental setup is sufficient to uniquely estimate the parameters of interest. If a system is not identifiable, then an identification algorithm may give incorrect parameter values without warning. Similarly, if one wishes to discriminate between two possible models for a system, the property of discriminability is necessary to guarantee \emph{a priori} that the true model can be determined from data.
If discriminability is not guaranteed then an algorithm that determines which model generated data can select the wrong model.  In the context of ordinary differential equation (ODE) models, identifiability analysis often takes the form of determining which set of input signals are sufficient to identify the parameters, while discriminability analysis takes the form of determining which input signals are sufficient to select the true model.

\emph{Global a priori} identifiability is the strongest type of structural identifiability, which guarantees that no matter what the true parameter values are, one will be able to uniquely determine them from a given experiment as long as sufficient data is gathered~\cite{Ljung:1994wl}. In general, proving that global identifiability holds is difficult~\cite{cinquemani2018identifiability,hsiao2018control}, and for ODE models a variety of computational tools have been developed. Some exploit the differential algebraic structure of the problem to analyze identifiability with Ritt's Algorithm~\cite{Ljung:1994wl,Bellu:2007tt,audoly2001global}, while other methods are based on observability analysis, with the parameters treated as states with trivial dynamics~\cite{sontag2017dynamic,sontag1979observability,sontag2013mathematical,Villaverde:2016vr,chis2011structural,villaverde2016identifiability}. 

Most work on identifiability for biological applications has focused on ODE models that describe the time evolution of the mean values of the state variables, using the previously discussed algorithmic tools. However, in biological applications,  common data include single cell measurements from a population of cells, such as obtained from flow cytometry~\cite{shapiro2005practical} or from single cell RNAseq~\cite{luecken2019current}. While these techniques can obtain measurements of population distributions across many cells, they do not allow tracking individuals cells across time. Therefore, the data does not take the form of (possibly noisy) measurements along a sample path of the system and thus the standard methods for identifiability analysis of dynamical systems are not directly applicable. However, it has been observed in a variety of studies that using information about the time evolution of the population distribution over the outputs can help identify more parameters than just the time evolution of the means of the outputs in specific cases~\cite{munsky2009listening,Munsky7533,Lillacci:2013wu,steuer2003observing}. Despite this, no general framework for identifiability analysis exists in this setting. When the time evolution of the population distribution can be described by a system of finitely many ODEs, methods of identifiability analysis for ODE models such as those in \cite{Ljung:1994wl} and \cite{Villaverde:2016vr} can be used. Cinquemani studied identifiability of chemical reaction networks from a sequence of distributional data~\cite{cinquemani2018identifiability}. However, their results are only valid for local identifiability of chemical reaction networks with propensities that are affine in the state, e.g., monomolecular reactions, and therefore these results do not allow analysis of general chemical reaction networks or of global identifiability. 

A special case of distributional data measures only the stationary distribution, i.e., just the equilibrium population distribution. In this scenario, algorithms to identify chemical reaction network parameters from stationary distributions have been developed \cite{gupta2019bayesian,ocal2019parameter,backenkohler2017moment}. However, none of these works considered the question of identifiability. Therefore, generally applicable methods for identifiability analysis when only the stationary distribution is measured have been lacking.
In fact, to the best of the authors' knowledge, the question of identifiability from only the stationary distribution has not been studied for general chemical reaction networks.  Swaminathan and Murray considered identifiability of linear time invariant systems from the stationary distribution over all states and additionally a sample path of the underlying stochastic process for a subset of states~\cite{swaminathan2016linear}, but they did not provide conditions for identifiability in the case of only distributional data.

An additional source of noise in biological systems is extrinsic noise. Extrinsic noise arises from the variability of cellular context across a population of cells~\cite{Swain:2002ww}. In this work we additionally consider extrinsic noise that manifests through parameter variation between cells in a population. Such noise can arise from a variety of sources, most notably in synthetic genetic circuits from differences in copy number of the DNA on which the genetic circuit is encoded, such as with lentiviral transduction in mammalian cells or with plasmid transfection in either bacterial or mammalian cells~\cite{campeau2009versatile,schwake2010predictive}. Such noise can, in principle, improve our ability to identify the reaction rate constants, since we have data across a wider range of conditions. However, this is not clear \emph{a priori}.

In this work, we consider global identifiability of linear noise approximation (LNA) models~\cite{van1992stochastic} of chemical reaction networks with intrinsic and extrinsic noise from their stationary distributions, including a treatment of the model discrimination case where one wishes to know if it is possible to determine which chemical reactions are present in a system. Our solution is a generally applicable algebraic characterization of identifiability, which is amenable to analysis using Hilbert's Nullstellensatz~\cite{cox2013ideals}, and thus allows the computation of certificates of identifiability.

This paper is organized as follows. In Section \ref{sec:problem_setting}, we give mathematical background and a description of the problem we consider. In Section \ref{sec:main_result}, we give the main results of this paper, describing how to use algebraic tools to certify global identifiability of chemical reaction networks from their stationary distributions. In particular, Section \ref{sec:main_result} describes a chemical reaction network modeled by the LNA where the goal is to identify the values of the reaction rate constants. In Section \ref{sec:idble_model_id}, we show how to approach the model discriminability problem using our techniques. In Section \ref{sec:idble:extrinsic}, we show how to certify global identifiability from the stationary distribution for chemical reaction networks with extrinsic and intrinsic noise, and additionally show that the addition of extrinsic noise cannot make an identifiable chemical reaction network non identifiable. Throughout this work we apply our methods to certify identifiability of a wide range of chemical reaction networks.

\section{Problem Setting}\label{sec:problem_setting}

\subsection{The linear noise approximation}\label{sec:LNA}
A chemical reaction network (CRN) is a model of a system of chemical species interacting through reactions, each of which is a discrete event that occurs stochastically. The exact model of the resulting stochastic kinetics is given by the chemical master equation, an infinite set of ordinary differential equations that describes the time evolution of the probability of having a particular number of molecules of each species in the system~\cite{gillespie1992rigorous}. In this work, we use the LNA as a model of the stochastic dynamics of CRNs. The LNA, also known as the system size expansion, is the first order correction to the  deterministic reaction rate equations in $\Omega^{-1/2}$, where $\Omega$ is the volume in which the chemical species are contained~\cite{van1992stochastic}. Letting $X$ represent the vector of molecular counts of each species, and $\vec{x}$ represent the mean concentration of the molecular species, the LNA makes the approximation $X = \Omega \vec{x} + \sqrt{\Omega} \vec{\xi}$. Here, $\vec{x}$ is the deterministic mean, which is given by the reaction rate equations, an ODE model that describes the rate of change of the molecular species concentrations, assuming mass action kinetics~\cite{gillespie1992rigorous}, and $\vec{\xi}$ is a random variable representing the fluctuations of $X$ about $\Omega\vec{x}$. For completeness, we give a brief description of the LNA here, a full derivation is given in \cite{van1992stochastic}. We remark that while the LNA gives distributions that are close to the distributions given by the chemical master equation when the volume and molecular counts are large on a finite time interval~\cite{kurtz1978strong}, there are no formal guarantees that the stationary distribution of the LNA is close to that of the chemical master equation. In this work, we take the stationary distribution of the LNA as our model of the stationary distribution of a CRN.

Consider a CRN consisting of $r$ reactions among $n$ species in a well mixed volume of size $\Omega$. Reaction $i$, for $i \in \{1,\dots,r\}$, is described by $\vec{s}_{ri}^T\vec{X} \xrightarrow{k_i} \vec{s}_{pi}^T\vec{X}$, where $\vec{X} = \begin{bmatrix} X_1 & X_2 & \cdots X_n\end{bmatrix}^T$ with $X_j$ the number of molecules of species $j$, $\vec{s}_{ri}$ is the vector of number of molecules of \emph{reactant} species consumed by reaction $i$, and $\vec{s}_{pi}$ is the vector of number of molecules of \emph{product} species created by reaction $i$. The reaction rate constant of reaction $i$ is $k_i$. Using the approximation $\vec{X}(t) = \Omega \vec{x}(t) + \sqrt{\Omega} \vec{\xi}(t)$, the dynamics of the system are given by
\begin{subequations}\label{eq:LNA}
\begin{align}
	\frac{d}{dt}\vec{x}(t) =& \vec{f}(\vec{x}(t);\vec{k}), \quad \vec{x}_0(0) = \vec{x}_0,\label{eq:LNA_rre}\\
	d\vec{\xi}(t) =& \frac{\partial \vec{f}}{\partial \vec{x}} \vec{\xi}(t) dt + \Gamma(\vec{x}(t);\vec{k}) d\vec{w}(t),\quad \vec{\xi}(0) = \vec{\xi}_0,\label{eq:LNA_fluc}
\end{align}
\end{subequations}\noeqref{eq:LNA_rre}\noeqref{eq:LNA_fluc}
in which \eqref{eq:LNA_rre} are the reaction rate equations (RRE) \cite{gillespie1992rigorous} and \eqref{eq:LNA_fluc} gives the evolution of $\vec{\xi}(t)$. Specifically, let $\vec{k} = [k_1,\dots,k_r]^T$. Then, $\vec{f}(\vec{x};\vec{k})$ is given by
\begin{equation}\label{eq:RRE_def}
	\vec{f}(\vec{x}; \vec{k})= S\vec{q}(\vec{x};\vec{k}) %
\end{equation}
where $\vec{q}(\vec{x};\vec{k}) = \begin{bmatrix} q_1(\vec{x};k_1) & q_2(\vec{x};k_2) & \cdots & q_r(\vec{x};k_r)\end{bmatrix}^T$, where $q_i(\vec{x};k_i) = k_i \prod_{j=1}^n x_j   ^{s^j_{ri}}$ is the macroscopic propensity of reaction $i$, where $s^j_{ri}$ is the $j$\textsuperscript{th} element of $\vec{s}_{ri}$.  The stoichiometry matrix $S$ is defined as $S = \begin{bmatrix} \vec{s}_1 & \vec{s}_2 & \cdots & \vec{s}_r\end{bmatrix}$, with $\vec{s}_i = \vec{s}_{pi} - \vec{s}_{ri}$ representing the change in $\vec{X}$ when reaction $i$ occurs. Here, $\vec{w}(t)$ is a Wiener process, and
\begin{equation}\label{eq:Gamma_def}
	\Gamma(\vec{x}; \vec{k})  = S\diag\left(\sqrt{\vec{q}(\vec{x};\vec{k})}\right).
\end{equation}
Throughout this work, we assume that \eqref{eq:LNA_rre} has a unique, exponentially stable, equilibrium in $\mathbb{R}_{\geq 0}^n$ for all $\vec{k} >0$. We denote this equilibrium point by $\vec{x}^*(\vec{k})$. Let $P\in \mathbb{R}^{n\times n}$ be the stationary covariance of $\vec{\xi}$. Then, the following equations characterize the stationary distribution of $\vec{X}(t)$ as a function of $\vec{k}$:
\begin{subequations}\label{eq:stationary_dist}
\begin{align}
	0 =& \vec{f}(\vec{x};\vec{k}),\label{eq:stationary_dist_rre}\\
	0 =& \frac{\partial \vec{f}}{\partial \vec{x}} P + P \frac{\partial \vec{f}}{\partial \vec{x}} ^T + \Gamma(\vec{x}; \vec{k})\Gamma(\vec{x}; \vec{k})^T.\label{eq:stationary_dist_covariance}
\end{align}
\end{subequations}
The stationary distribution of $\vec{X}(t)/\Omega$ is $\mathcal{N}(\vec{x}^*(\vec{k}), \frac{1}{\Omega}P^*(\vec{k}))$, i.e., a normal distribution with mean $\vec{x}^*(\vec{k})$ and covariance $\frac{1}{\Omega}P^*(\vec{k})$, where $\vec{x}^*(\vec{k})$ and $P^*(\vec{k})$ are the solutions to \eqref{eq:stationary_dist}. Our assumption that \eqref{eq:LNA_rre} has a unique equilibrium point in $\mathbb{R}^n_{\geq0}$ for all $\vec{k} >0$ ensures that \eqref{eq:stationary_dist} defines the unique stationary distribution under the LNA. For brevity, we denote a CRN as a function $\mathcal{R}$ that maps reaction rate vectors to the corresponding stationary distribution according to \eqref{eq:stationary_dist}, i.e., $\mathcal{R}:\mathbb{R}_{>0}^r \rightarrow \mathbb{R}^n \times \mathbb{S}^{n\times n}$, where $\mathbb{S}^{n\times n}$ is the space of symmetric $n \times n$ real matrices, defined by $\mathcal{R}(\vec{k}) = \left(\vec{x}^*(\vec{k}),\frac{1}{\Omega}P^*(\vec{k}\right))$.

\begin{example}[Illustrative Example 1]\label{ex:one_dim} 
We first consider a simple CRN $\mathcal{R}_1$ with a single species ($n=1$) and three reactions ($r=3$) given by
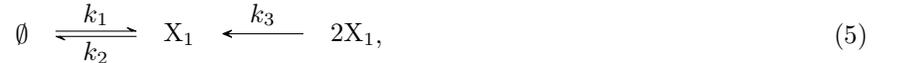
\begin{equation}\label{eq:ex1_reactions}
	\schemestart
		\subscheme{$\emptyset$}
    		\arrow(z--x1){<=>[$k_1$][$k_2$]}[0]
		\subscheme{$\mathrm{X}_1$}
		\arrow(@x1--x11){<-[$k_3$]}[0]
		\subscheme{$2\mathrm{X}_1$}
	\schemestop ,
\end{equation}
where reaction $i$ is labeled with its reaction rate constant, $k_i$. The reaction rate equation \eqref{eq:RRE_def} in this case given by
\begin{equation}\label{eq:ex1_f}
	\frac{d}{dt}x_1 = \vec{f}(\vec{x};\vec{k}) = k_1 - k_2x_1 - k_3x_1^2,
\end{equation}
from which we see that there is a unique and asymptotically stable equilibrium point in the region $x_1 \geq 0$ as long as $\vec{k}>0$, and thus the LNA model has a unique equilibrium distribution. In this case we have $\vec{q}(\vec{x};\vec{k}) = \begin{bmatrix} k_1 & k_2x_1 & k_3x_1^2\end{bmatrix}^T$ and the stoichiometry matrix is $S = \begin{bmatrix} 1 & -1 & -1 \end{bmatrix}$. Therefore, from \eqref{eq:Gamma_def} we have
\begin{equation}\label{eq:ex1_gamma}
	\Gamma(x; \vec{k})\Gamma(x; \vec{k})^T = k_1 + k_2x_1 + k_3x_1^2.
\end{equation}
\end{example}

\subsection{Identifiability}\label{sec:identifiability}

In this work, we study the following problem: Given $\vec{\pi}^*$, a stationary distribution over the species concentrations, and $K\subseteq \mathbb{R}^r_{>0}$ a set of possible $\vec{k}$ values, can we uniquely identify the $\vec{k}$ which gave rise to $\vec{\pi}^*$? To make this question mathematically precise, we will consider the following definition of global identifiability for CRNs from the stationary distribution.
\begin{definition}\label{def:param_id}
	A CRN $\mathcal{R}(\vec{k})$ is \emph{stationary globally identifiable} over $K\subseteq \mathbb{R}^r_{>0}$ if for any $\vec{k}_1,\vec{k}_2\in K$ such that $\mathcal{R}(\vec{k}_1) = \mathcal{R}(\vec{k}_2)$, there exists $a \in \mathbb{R}$ such that $\vec{k}_2 = a\vec{k}_1$.
\end{definition}
\begin{remark}
 	For any CRN, if one scales all of the reaction rate constants by the same value, $a$, the stationary distribution does not change. This fundamental lack of identifiability is due to our inability to tell the `speed' of a continuous time Markov chain from its stationary distribution. Definition \ref{def:param_id} reflects that fact that here we study identifiability modulo this fundamental source of non-identifiability.
\end{remark}
\begin{remark}
Whether or not a system is identifiable depends entirely on the model, which is given by the LNA in our analysis. However, under certain conditions, the first and second moments of the LNA and chemical master equation models are identical~\cite{grima2015linear-noise}, and hence in those cases our results also imply identifiability of the chemical master equation model.
\end{remark}

\subsection{Nullstellensatz}\label{sec:nullstellensatz}
In this section, we briefly describe the algebraic tools that we use in this work~\cite{cox2013ideals}. Let $\vec{z}$ be an $n'$ dimensional vector of variables. We denote the set of polynomials in $\vec{z}$,with rational coefficients by $\mathbb{Q}[\vec{z}]$. Since $p \in \mathbb{Q}[\vec{z}]$ is a function of $\vec{z}$, for any $\vec{z}'\in\mathbb{C}^{n'}$, $p(\vec{z}')$ denotes $p$ evaluated at $\vec{z}'\in\mathbb{C}^{n'}$. We say that $p \in \mathbb{Q}[\vec{z}]$ is a monomial if $p$ can be written as $p = \prod_{i=1}^N z_i^{\alpha_i}$ for some $N \geq 0$ and $\alpha_1,\alpha_2,\dots ,\alpha_N \in \mathbb{N}$. Let ``$\prec$'' be any total ordering \cite{cox2013ideals} on the set of monomials in $\mathbb{Q}[\vec{z}]$ that additionally satisfies \emph{i)} $1 \prec p$ for any nonconstant monomial $p\in \mathbb{Q}[\vec{z}]$ and \emph{ii)} $\prod_{i=1}^N x_i^{\alpha_i} \prec  \prod_{i=1}^N z_i^{\beta_i}$ implies that $ \prod_{i=1}^N z_i^{\alpha_i + \gamma_i} \prec  \prod_{i=1}^N z_i^{\beta_i + \gamma_i}$ for all $\alpha_1,\dots,\alpha_N,\beta_1,\dots,\beta_N,\gamma_1,\dots,\gamma_N \in \mathbb{N}$. Such a total ordering $\prec$ is called a \emph{term order} on $\mathbb{Q}[\vec{z}]$. The \emph{ideal} generated by a  set of polynomials $\mathcal{P} \subseteq \mathbb{Q}[\vec{z}]$ is defined as all polynomial combinations of the elements of $\mathcal{P}$, i.e.,
\begin{equation}
	 \left\langle \mathcal{P}\right\rangle = \left\{ g \in \mathbb{Q}[\vec{z}] \middle| g = \sum_{i=1}^m \lambda_i p_i,\; p_1,\dots ,p_m \in \mathcal{P},\;\lambda_1,\lambda_2,\dots,\lambda_m \in \mathbb{Q}[\vec{z}], \text{for some $m \in \mathbb{N}$} \right\}.
\end{equation}

\setcounter{example}{1}
\begin{example}[Algebraic preliminaries]\label{ex:algebraic}
	To illustrate the concepts we consider two different sets of polynomials, $\mathcal{P}_1 = \{z^2 - 1,\;z-1\} \subset \mathbb{Q}[z]$ and $\mathcal{P}_2 = \{z^2-1,\; z-2\}\subset \mathbb{Q}[z]$. We have that 
	\begin{equation}
		\left\langle \mathcal{P}_1\right\rangle = \left\{ g \in \mathbb{Q}[z] \middle| g = \lambda_1\left(z^2 - 1\right) + \lambda_2\left(z-1\right),\; \lambda_1,\lambda_2 \in \mathbb{Q}[z] \right\}
	\end{equation}
	and
	\begin{equation}
		\left\langle\mathcal{P}_2\right\rangle = \left\{ g \in \mathbb{Q}[z] | g = \lambda_1\left(z^2 - 1\right) + \lambda_2\left(z-2\right),\; \lambda_1,\lambda_2 \in \mathbb{Q}[z]\right\}.
	\end{equation}
	For example, $\mathcal{P}_1$ contains $0$ (with $\lambda_1 = 0$, $\lambda_2 = 0$), $z^2-1$ (with $\lambda_1 = 1$, $\lambda_2 = 0$), $z-1$ (with $\lambda_1 = 0$, $\lambda_2 = 1$), as well as $z^3 - 1$ (with $\lambda_1 = z$, $\lambda_2 = 1$), but does not contain $1$, since no $\lambda_1,\lambda_2 \in \mathbb{Q}[z]$ results in $1 = \lambda_1(z^2-1) + \lambda_2(z-1)$. On the other hand, $\mathcal{P}_2$ does contain $1$, since $\lambda_1 = 2z/3 -1$ and $\lambda_2 = -2z^2/3 - z/3$ results in $\lambda_1(z^2-1) + \lambda_2(z-2) = 1$.
\end{example}

Let $p\in \mathbb{Q}[\vec{z}]$. Then, $\mathit{\mathit{in}_{\prec}(p)}$ denotes the largest monomial with respect to $\prec$ that appears in $p$  with a nonzero coefficient. Suppose $\mathcal{I} = \left\langle \mathcal{P}\right\rangle$, then $\mathcal{G}$ is a Gr\"{o}bner basis of $\mathcal{I}$ if it is a finite subset of $\mathcal{I}$ that satisfies $\left\langle\mathrm{in}_{\prec}(p)\middle| p \in \mathcal{I}\right\rangle = \left\langle\mathrm{in}_{\prec}(g) \middle| g \in \mathcal{G}\right\rangle$. $\mathcal{G}$ is a reduced Gr\"{o}bner basis of $\mathcal{I}$ if additionally \emph{i)} the coefficient of the largest monomial in $g$ with respect to $\prec$ is $1$ for each $g \in \mathcal{G}$ and \emph{ii)} for all $g\in\mathcal{G}$, $\left\langle\mathrm{in}_{\prec}(g')\middle| g' \in \mathcal{G} \setminus \{g\}\right\rangle$ does not contain any monomial term of $g$. 
 In Example \ref{ex:algebraic} and for the rest of this work we use Buchberger's algorithm, as implemented in Macaulay2, to compute reduced Gr\"{o}bner bases~\cite{bose1995grobner,M2}.
\setcounter{example}{1}
\begin{example}[Algebraic preliminaries continued]
	 Continuing Example \ref{ex:algebraic}, we consider the reduced Gr\"{o}bner bases of $\mathcal{P}_1$ and $\mathcal{P}_2$. When $n' = 1$, the only valid term order is $1 \prec z \prec z^2 \prec \dots$. The reduced Gr\"{o}bner basis of $\langle \mathcal{P}_1 \rangle$ is $\mathcal{G}_1 = \{z-1\}$ with respect to this term order, whereas with the same term order the reduced Gr\"{o}bner basis of $\mathcal{P}_2$ is $\{1\}$. The details of computing reduced Gr\"{o}bner bases can be found in \cite{cox2013ideals}.
\end{example}

Given an ideal $\mathcal{I} = \langle \mathcal{P}\rangle$, there are many sets of polynomials that generate $\mathcal{I}$. The reduced Gr\"{o}bner basis is a special choice of generating polynomials which reveals certain properties of $\mathcal{I}$. In particular, let $\mathcal{V}(\mathcal{P})$ denote the variety of $\mathcal{P}$, defined by
\begin{equation}
	\mathcal{V}(\mathcal{P}) = \left\{\vec{z} \in \mathbb{C} \middle| 0 = p(\vec{z}),\; \forall p \in \mathcal{P}\right\}.
\end{equation}
In other words if $\mathcal{P} = \left\{p_1,p_2,\dots, p_m\right\}$, $\mathcal{V}(\mathcal{P})$ is the set of solutions to the system of equations $0 = p_1(\vec{z}),\,0 = p_2(\vec{z}),\dots,0 = p_m(\vec{z})$. It is true that $\mathcal{V}(\mathcal{P}) = \mathcal{V}(\mathcal{F})$ for any $\mathcal{F}$ such that $\mathcal{I} = \langle \mathcal{F}\rangle$. In particular, if $\mathcal{G}$ is a reduced Gr\"{o}bner basis of $\mathcal{I}$, then $\mathcal{V}( \mathcal{P}) = \mathcal{V}(\mathcal{G})$. Therefore, if we wish to study $\mathcal{V}(\mathcal{P})$, the set of common zeros of the polynomials in $\mathcal{P}$, we can study $\mathcal{V}(\mathcal{G})$ instead, which is advantageous since by examining the reduced Gr\"{o}bner basis, one can easily tell if $\mathcal{V}(\mathcal{P})$ is empty or not. This idea is formalized by Hilbert's Nullstellensatz, one version of which is given here.
\begin{theorem}[See e.g. \cite{sturmfels2005grobner}]\label{thm:nullstellensatz_coro}
	Let $p_1,p_2,\dots,p_m \in \mathbb{Q}[\vec{z}]$ be polynomials in the $n'$ variables in $\vec{z}$. Then
	\begin{equation}
		\emptyset = \left\{\vec{z} \in \mathbb{C}^{n'} \middle| 0=p_1(\vec{x}),0=p_2(\vec{x}),\dots,0=p_m(\vec{x})\right\}
	\end{equation}
	if and only if the reduced Gr\"{o}bner basis of $\left\langle p_1, p_2, \dots, p_m\right\rangle$ is $\{1\}$.
\end{theorem}
\setcounter{example}{1}
\begin{example}[Algebraic preliminaries continued]
	Since the reduced Gr\"{o}bner basis of $\mathcal{P}_1$ is not $\{1\}$, from Theorem \ref{thm:nullstellensatz_coro} we can conclude that there is a solution in $\mathbb{C}$ to
	\begin{align}
		0 & = z^2 - 1,\\
		0 & = z - 1.
	\end{align}
	In fact, one can see that there is one solution, $z = 1$. On the other hand, the reduced Gr\"{o}bner basis of $\mathcal{P}_2$ is $\{1\}$ and therefore, from Theorem \ref{thm:nullstellensatz_coro}, we can conclude that there are no solutions in $\mathbb{C}$ to 
	\begin{align}
		0 & = z^2 - 1,\label{eq:ex_alg_P2_a}\\
		0 & = z - 2,\label{eq:ex_alg_P2_b}
	\end{align}
	which is consistent with our ability in this simple case to deduce that the sets of solutions to \eqref{eq:ex_alg_P2_a} and \eqref{eq:ex_alg_P2_b} do not intersect.
\end{example}

\section{Certifying Identifiability of the LNA}\label{sec:main_result}
We now present the main results of this work, which are methods to algorithmically test for stationary global indentifiability. We begin by showing that the right-hand side of \eqref{eq:stationary_dist} is linear in $\vec{k}$. Specifically, we can write \eqref{eq:stationary_dist_rre} as
\begin{equation}
	\vec{f}(\vec{x};\vec{k}) = \sum_{i=1}^r k_i \vec{s}_i \prod_{j=1}^n x_j^{s_{ri}^j},
\end{equation}
and, given \eqref{eq:Gamma_def}, \eqref{eq:stationary_dist_covariance} can be written as
\begin{equation}
	0 =  \frac{\partial \vec{f}}{\partial \vec{x}}P + P \frac{\partial \vec{f}}{\partial \vec{x}}^T + S\diag \vec{q}(\vec{x};\vec{k}) S^T,
\end{equation}
where we have used the fact that for all $\vec{x}\in\mathbb{R}_{\geq 0}^n$, it is true that $\vec{q}(\vec{x};\vec{k}) \geq 0$. Therefore, the right-hand side of \eqref{eq:stationary_dist_rre} is linear in $\vec{k}$. Furthermore, since $\frac{\partial \vec{f}}{\partial \vec{x}}$ and $\vec{q}(\vec{x};\vec{k})$ are linear in $\vec{k}$, the right-hand side of \eqref{eq:stationary_dist_covariance} is also linear in $\vec{k}$. Also, \eqref{eq:stationary_dist} give $n + n^2$ equations for $\vec{x}\in \mathbb{R}_{\geq}^n$ and $P\in \mathbb{S}^{n\times n}$. Since $P$ is symmetric, there are only $\frac{n^2 + n}{2}$ unique equations in \eqref{eq:stationary_dist_covariance}. Therefore, combining our observations about linearity and the number of unique equations, \eqref{eq:stationary_dist} can be written in the form
\begin{equation}\label{eq:stationary_A}
	0 = A(\vec{x},P)\vec{k},
\end{equation}
where $A(\vec{x},P) \in \mathbb{R}^{\frac{n^2+n}{2} \times r}$ is a function of $\vec{x}$ and of the $\frac{n^2 + n}{2}$ entries of $P$ that are on or above the diagonal. Additionally, since $\vec{f}(\vec{x};\vec{k})$ and $q_i(\vec{x};k_i)$ are polynomials in $\vec{x}$, the elements of $A(\vec{x},P)$ are polynomials in $\vec{x}$ and in the elements of $P$ on or above the diagonal.
\setcounter{example}{0}
\begin{example}[Illustrative example 1 continued]
We ask if $\mathcal{R}_1$, given by \eqref{eq:ex1_reactions}, is stationary globally identifiable over $\mathbb{R}^3_{>0}$. In this example, letting $\vec{x} = x_1$ and $P = p_{11}$, writing out \eqref{eq:stationary_dist} explicitly using \eqref{eq:ex1_f} and \eqref{eq:ex1_gamma} yields
\begin{subequations}\retainlabel{eq:example1_stationary}
\noeqref{eq:example1_stationary_rre,eq:example1_stationary_covariance}
\begin{align}
	0 & = k_1 - k_2x_1 - k_3x_1^2,\retainlabel{eq:example1_stationary_rre}\\
	0 & = 2(-k_2 - 2k_3x_1)p_{11} + k_1 + k_2x_1 + k_3x_1^2.\retainlabel{eq:example1_stationary_covariance}
\end{align}
\end{subequations}
We can write \eqref{eq:example1_stationary} as $0 = A(x,P)\vec{k}$ where
\begin{equation}\label{eq:ex1_A}
	A(\vec{x}, P) = \begin{bmatrix}
		1 & -x_1 & - x_1^2 \\
		1 & x_1 - 2p_{11} & x_1^2 - 4p_{11}x_1
	\end{bmatrix}.
\end{equation}
\end{example}

In general, proving that a given system is stationary globally identifiable is difficult, since it requires proving that \eqref{eq:stationary_A} has only one subspace of solutions in $\vec{k}$ for all $(\vec{x},P)$ that are feasible, that is, for all $(\vec{x},P)$ such that there exists $\vec{k}\in K$ satisfying $(\vec{x},P) = \mathcal{R}(\vec{k})$. These feasible $(\vec{x},P)$ are given by \eqref{eq:stationary_dist}, which is a set of polynomial equations in $(\vec{x},P)$, along with the constraint $\vec{k}\in K$. To overcome this difficulty, we develop a method to certify global stationary identifiability based on Theorem \ref{thm:nullstellensatz_coro}. To begin, associated with each CRN $\mathcal{R}$, we define the sets
\begin{equation}\label{eq:set_S_eq_zero}
	V = \left\{ (\vec{x}, P, \vec{k}) \in (\mathbb{R}^n, \mathbb{S}^{n \times n}, \mathbb{R}_{>0}^r) \middle| 0 = A(\vec{x},P)\vec{k},\; \mbox{rank}(A(\vec{x},P)) < r-1 \right\}.
\end{equation}
and
\begin{equation}\label{eq:set_S_eq_zero_converse}
	V' = \left\{ (\vec{x}, P, \vec{k}) \in (\mathbb{R}^n_{\geq 0}, \mathbb{S}^{n \times n}, \mathbb{R}_{>0}^r) \middle| 0 = A(\vec{x},P)\vec{k},\; \mbox{rank}(A(\vec{x},P)) < r-1 \right\}.
\end{equation}

The following theorem gives an algebraic characterization of stationary globally identifiable for a CRN.

\begin{theorem}\label{thm:spi_alg}
	Consider a CRN $\mathcal{R}$. The following hold:
	\begin{enumerate}
		\item[i)] If $V = \emptyset$, then $\mathcal{R}$ is stationary globally identifiable over $\mathbb{R}^r_{>0}$.
		\item[ii)] If $\mathcal{R}$ is stationary globally identifiable over $\mathbb{R}^r_{>0}$, then $V' = \emptyset$.
	\end{enumerate}
\end{theorem}
\begin{proof}
	First, to show i), suppose that $\mathcal{R}$ is not stationary globally identifiable over $\mathbb{R}^r_{>0}$. Then there exists $\vec{k}_1,\vec{k}_2 >0$, with $\vec{k}_2$ and $\vec{k}_1$ linearly independent, such that $0 = A(\vec{x},P)\vec{k}_1$ and $0 = A(\vec{x},P)\vec{k}_2$. This immediately implies that $\rank A(\vec{x},P) < r-1$, and therefore $(\vec{x},P,\vec{k}_1) \in V$. Now, to show ii), suppose that there exists $(\vec{x}',P',\vec{k}') \in V'$. By the definition of $V'$, $\rank A(\vec{x}',P') < r-1$, so there exists $W$, a subspace of dimension 2 containing $\vec{k}$ such that $0 = A(\vec{x}',P')W$. It then follows from the fact that $\mathbb{R}^r_{>0}$ is open that there exists $\vec{k}'' > 0$, linearly independent from $\vec{k}'$, such that $0 = A(\vec{x},P)\vec{k}''$. By the uniqueness of the equilibrium point of \eqref{eq:LNA_rre} in $\mathbb{R}_{\geq0}^n$, we know that $(\vec{x}',P')$ is the stationary distribution of $\mathbb{R}$ for all $\vec{k}\in W$, and therefore $\mathcal{R}$ is not stationary globally identifiable.
\end{proof}
\begin{remark}
	While our assumption that \eqref{eq:LNA_rre} has a unique, exponentially stable, equilibrium point in $\mathbb{R}_{\geq 0}^n$ is required for statement ii) of Theorem \ref{thm:spi_alg} to hold, this assumption is not required for statement i) of Theorem \ref{thm:spi_alg}.
\end{remark}

In the remainder of this section, we transform the rank condition on $A$ into a polynomial condition so that the question of the emptiness of $V$ can be addressed by algebraic techniques. To this end, we require the following Lemmas.

\begin{lemma}(Determinant rank characterization)\label{lem:det_rank_char_hornandj}
	Let $A \in \mathbb{R}^{n \times m}$. Then, $\rank A = r'$ if and only if every $r'+1 \times r'+1$ minor of $A$ is zero, and there exists an $r' \times r'$ minor of $A$ that is non-zero.
\end{lemma}
\begin{proof}
	See \cite[Section 0.4]{horn2012matrix}.
\end{proof}

\begin{lemma}\label{lem:det_rank_char_ineq}
	Let $A \in \mathbb{R}^{n \times m}$. Then, $\rank A < r'$ if and only if every $r' \times r'$ minor of $A$ is zero.
\end{lemma}
\begin{proof}
	First, we show that if $\rank A < r'$, then every $r' \times r'$ minor of $A$ is zero. Let $\rank A = r'' < r'$. Then, by Lemma \ref{lem:det_rank_char_hornandj}, every $r''+1\times r'' + 1$ minor of $A$ is zero. Furthermore, by the Laplace expansion for the determinant \cite{horn2012matrix}, for all $r''' \geq r'' + 1$, every $r''' \times r'''$ minor of $A$ is zero. Specifically, since $r' \geq r'' + 1$, every $r' \times r'$ minor of $A$ is zero. Second, we show that if $\rank A \geq r'$, then there exists a nonzero $r'\times r'$ minor of $A$. Let $\rank A = r'' \geq r'$. By Lemma \ref{lem:det_rank_char_hornandj} there exists an $r''\times r''$ nonzero minor of $A$. It follows from the Laplace expansion for the determinant \cite{horn2012matrix} that for all $r''' \leq r''$ there exists an $r'''\times r'''$ nonzero minor of $A$. Specifically, there exists an $r' \times r'$ nonzero minor of $A$.
\end{proof}
We now use Lemma \ref{lem:det_rank_char_ineq} and Theorem \ref{thm:spi_alg} to give a computationally checkable sufficient condition for a CRN to be stationary globally identifiable.
\begin{theorem}\label{thm:grob_check}
	Consider a CRN $\mathcal{R}$. If the reduced Gr\"{o}bner basis of
	\begin{equation}\label{eq:ideal}
	\begin{multlined}
		\mathcal{I} = \left\langle y_j^2k_j - 1\; \forall j\in\{1,\dots,r\},\;A_q(\vec{x},P)\vec{k}\; \forall q \in \{1,\dots r\},\vphantom{M^{(r)}}\right.\\\left. M^{(r-1)\times (r-1)}_i(\vec{x},P)\;  \forall i \in \{1,\dots,m\} \right\rangle
	\end{multlined}
	\end{equation}
	is $\{1\}$, then $\mathcal{R}$ is stationary globally identifiable over $\mathbb{R}^r_{>0}$. Here, $A_q(\vec{x},P)$ is the $q$\textsuperscript{th} row of $A(\vec{x},P)$ and $M^{(r-1)\times (r-1)}_i(\vec{x},P)$ is all of the size $(r-1) \times (r-1)$ minors of $A(\vec{x},P)$, indexed by $i=1,\dots,m$.
\end{theorem}

\begin{remark}
	The ideal $\mathcal{I}$ defined in \eqref{eq:ideal} is a subset of $\mathbb{Q}[(\vec{x},\vec{y},\vec{k})]$.
\end{remark}

\begin{proof}
Let
	\begin{multline}
		\bar{V} = \left\{ (\vec{x}, P, \vec{k}, \vec{y}) \in (\mathbb{R}^n, \mathbb{S}^{n \times n}, \mathbb{R}^r, \mathbb{R}^r) \middle| 0 = A(\vec{x},P)\vec{k},\vphantom{M^{(r)}}\right.\\\left. 0 = M^{(r-1)\times (r-1)}_i(\vec{x},P)\; \forall i \in \{1,\dots,m\},\; 0 = y_j^2k_j - 1\; \forall j\in \{1,\dots,r\} \right\}.
	\end{multline}
 Recall $V$ defined in \eqref{eq:set_S_eq_zero}. We first show that $V = \emptyset$ if and only if $\bar{V} = \emptyset$. First, suppose $V \neq \emptyset$. Then, there exists $(\vec{x}, P, \vec{k}) \in V$. It follows that $0 = A(\vec{x},P)\vec{k}$. Let $\vec{y}$ be such that $y_j = \sqrt{1/k_j}$. Therefore, for all $j$, $y_j^2k_j - 1 = 0$. By Lemma \ref{lem:det_rank_char_ineq}, $\mbox{rank}(A(\vec{x},P)) < r-1$ guarantees that $0 = M^{(r-1)\times (r-1)}_i(\vec{x},P)$ for all $i = 1,\dots,m$, and hence $(\vec{x}, P, \vec{k}, \vec{y}) \in \bar{V}$. 
Now suppose that $\bar{V} \neq \emptyset$. Then, there exists $(\vec{x}, P, \vec{k}, \vec{y}) \in \bar{V}$. It follows that $0 = A(\vec{x},P)\vec{k}$. Then, we have that $0 = M^{(r-1)\times (r-1)}_i(\vec{x},P)$  for all $i = 1,\dots,m$, and hence by Lemma \ref{lem:det_rank_char_ineq} it is true that $\rank A(\vec{x},P) < r-1$. Therefore $(\vec{x}, P, \vec{k}) \in V$, and hence $V \neq \emptyset$.
To complete the proof, observe that $\bar{V}$ is the variety of $\mathcal{I}$ defined by \eqref{eq:ideal}. If the reduced Gr\"{o}bner basis of $\mathcal{I}$ is $\{1\}$ then by Theorem \ref{thm:nullstellensatz_coro} $\bar{V} = \emptyset$. This implies by our above argument that $V = \emptyset$, and therefore by Theorem \ref{thm:spi_alg} $\mathcal{R}$ is stationary globally identifiable over $\mathbb{R}^r_{>0}$.
\end{proof}

Since the computation of reduced Gr\"{o}bner bases can be done algorithmically, Theorem \ref{thm:grob_check} allows us to check if a CRN is stationary globally identifiable automatically.
\begin{remark}\label{remark:positivstellensatz}
	Even though in this work we focus on using Hilbert's Nullstellensatz to certify identifiability, alternatively Positivstellensatz can be used to search for a certificate that $V=\emptyset$~\cite{stengle1974nullstellensatz}.
\end{remark}

\setcounter{example}{0}
\begin{example}[Illustrative example 1 continued]
We continue with Example 1. We ask if $\mathcal{R}_1$, given by \eqref{eq:ex1_reactions}, is stationary globally identifiable over $\mathbb{R}^3_{>0}$. In this case, $r=3$, $n=1$, $\vec{x}= x_1$, and $P = p_{11}$. Using \eqref{eq:ex1_A}, \eqref{eq:ideal} becomes
\begin{equation}\label{eq:example1_ideal}
\begin{multlined}
	\left\langle k_1y_1^2 - 1, k_2y_2^2 - 1, k_3y_3^2 - 1,\;
	k_1 - k_2x_1 - k_3x_1^2,\right.\\\left.
	k_1 - k_3(4p_{11}x_1 - x_1^2) - k_2(2p_{11} - x_1),\;
	 2x_1 - 2p_{11}, 2x_1^2 - 4p_{11}x_1, 2p_{11}x_1^2\right\rangle.
\end{multlined}
\end{equation}
Computing the reduced Gr\"{o}bner basis of \eqref{eq:example1_ideal} using the built in implementation of Buchberger's algorithm in Macaulay2 \cite{M2}, we find that it is $\{1\}$~\cite{M2}. Therefore, by Theorem \ref{thm:grob_check}, $\mathcal{R}_1$ is stationary globally identifiable over $\mathbb{R}^3_{>0}$.%
\end{example}

\subsection{Examples}
In this section, we present several examples of using the mathematical tools of Section \ref{sec:main_result} to certify that a given CRN is stationary globally identifiable. For all of the examples in this section, we compute reduced Gr\"{o}bner bases with Macaulay2, a software system for algebraic geometry \cite{M2}.

\setcounter{example}{2}
\begin{example}[Two species illustrative example]\label{ex:complex_balanced}
We now consider CRN $\mathcal{R}_3$ shown in \eqref{eq:complex_balanced}:
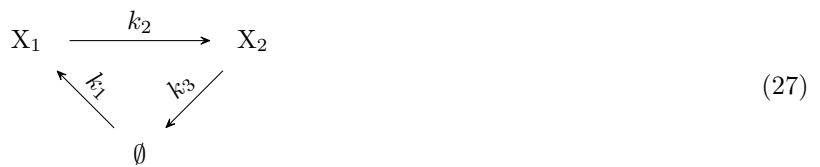
\begin{equation}\label{eq:complex_balanced}
	\vcenter{\hbox{\schemestart
		\subscheme{$\emptyset$}
    		\arrow(z--x1){->[$k_1$]}[135]
		\subscheme{$\mathrm{X}_1$}
		\arrow(@z--x2){<-[$k_3$]}[45]
		\subscheme{$\mathrm{X}_2$}
		\arrow(@x1--@x2){->[$k_2$]}
	\schemestop}}
\end{equation}
$\mathcal{R}_3$ has two species, $\mathrm{X}_1$ an $\mathrm{X}_2$. $\mathrm{X}_1$ is produced with rate constant $k_1$ and spontaneously transforms into $\mathrm{X}_2$ with rate constant $k_2$, which is degraded with rate constant $k_3$. We wish to understand if it is possible to estimate the rate vector $\vec{k}$ up to a scaling factor from the stationary distribution. For this example, $\vec{f}(\vec{x};\vec{k})$ defined in \eqref{eq:RRE_def} is
\begin{equation}\label{eq:ex3_rre}
	\vec{f}(\vec{x};\vec{k}) = 
	\begin{bmatrix}
		k_1 - k_2x_1\\
		k_2x_1 - k_3x_2
	\end{bmatrix},
\end{equation}
and $\Gamma(\vec{x};\vec{k})$ defined in \eqref{eq:Gamma_def} is
\begin{equation}\label{eq:ex3_gamma}
	\Gamma(\vec{x};\vec{k})\Gamma(\vec{x};\vec{k})^T =
	\begin{bmatrix}
		k_1 + k_2x_1 & -k_2x_1\\
	  -k_2x_1 & k_2x_1 + k_3x_2
	\end{bmatrix}.
\end{equation}
Writing \eqref{eq:LNA} in the form \eqref{eq:stationary_A} yields
\begin{equation}\label{eq:ex3_A}
	0 = A(\vec{x},P)\vec{k} = 
	\begin{bmatrix}
		1 &         -x_1 & 0 &           0 &           0\\
		0 &           0 & 1 &         -x_2 &          x_1\\
		1 & x_1 - 2p_{11} & 0 &           0 &           0\\
		0 &       -p_{12} & 0 &       -p_{12} &        p_{11}\\
		0 &           0 & 1 & x_2 - 2p_{22} & 2p_{12} + x_1
	\end{bmatrix}\vec{k}.
\end{equation}
Computing the reduced Gr\"{o}bner basis $\mathcal{G}$ of the ideal defined by \eqref{eq:ideal} with $A$ given in \eqref{eq:ex3_A}, we find that $\mathcal{G}=\{1\}$, and hence by Theorem \ref{thm:grob_check} $\mathcal{R}_3$ is stationary globally identifiable over $\mathbb{R}^3_{>0}$.%
\end{example}

\begin{example}[Sequestration rate]\label{ex:seq_rate}
Consider a CRN $\mathcal{R}_4$ consisting of two species $\mathrm{X}_1$ and $\mathrm{X}_2$ as shown in \eqref{eq:ex4_reactions}:
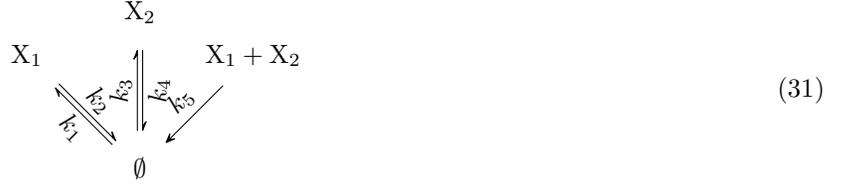
\begin{equation}\label{eq:ex4_reactions}
	\vcenter{\hbox{\schemestart
		\subscheme{$\emptyset$}
    		\arrow(z--x1){<=>[$k_2$][$k_1$]}[135]
		\subscheme{$\mathrm{X}_1$}
		\arrow(@z--x2){<=>[$k_3$][$k_4$]}[90]
		\subscheme{$\mathrm{X}_2$}
		\arrow(@z--x12){<-[$k_5$]}[45]
		\subscheme{$\mathrm{X}_1 + \mathrm{X}_2$}
	\schemestop}}
\end{equation}
Each species is produced and degraded at some unknown rate, and additionally $\mathrm{X}_1$ and $\mathrm{X}_2$ mutually degrade through the reaction \schemestart $\mathrm{X}_1+\mathrm{X}_2$\arrow{->[$k_5$]}$\emptyset$\schemestop. Such a system of chemical reactions is referred to as the \emph{antithetic} motif, and can be used to realize an integral controller~\cite{qian2018realizing,huang2018quasi,Aoki:2019aa}. Controllers constructed using the antithetic motif only approximately implement an integrator~\cite{qian2018realizing}. Based on \cite{qian2018realizing}, we can establish a heuristic to compare two possible biological implementations of the antithetic motif with parameter vectors $\vec{k}^A$ and $\vec{k}^B$ respectively with respect to the steady state error generated in a feedback system. To do this, we define the following dimensionless parameters:
\begin{align}\label{eq:sigmas}
		\sigma_1\left(\vec{k}^A,\vec{k}^B\right) &= \textstyle\frac{k_2^B k_5^A}{k_5^Bk_2^A}, & \sigma_2\left(\vec{k}^A,\vec{k}^B\right) &= \textstyle\frac{k_2^Bk_1^A}{k_1^Bk_2^A},\\ \sigma_3\left(\vec{k}^A,\vec{k}^B\right) &= \textstyle\frac{k_4^Bk_5^A}{k_5^Bk_4^A}, & \sigma_4\left(\vec{k}^A,\vec{k}^B\right) &= \textstyle\frac{k_4^Bk_3^A}{k_3^Bk_4^A}.
\end{align}
If $\sigma_i(\vec{k}^A,\vec{k}^B) << 1$ for $i \in \{1,2,3,4\}$, then $\vec{k}^B$ is expected to perform better than $\vec{k}^A$.  We observe that 
for all $\alpha^A,\alpha^B > 0$ we have $\sigma_i\left(\alpha^A\vec{k}^A,\alpha^B\vec{k}^B\right) = \sigma_i\left(\vec{k}^A,\vec{k}^B\right)$ for $i\in\{1,2,3,4\}$. Therefore, stationary global identifiability ensures that one can estimate $\sigma_i\left(\vec{k}^A,\vec{k}^B\right)$ for $i=1,2,3,4$ from the stationary distribution of $\mathcal{R}_4$. Motivated by this we study whether $\mathcal{R}_4$ is stationary globally identifiable. For $\mathcal{R}_4$ we have that 
\begin{equation}
	\vec{f}(\vec{x};\vec{k}) = 
	\begin{bmatrix}
		k_1 - k_2x_1 - k_5x_1x_2\\
		k_3 - k_4x_2 - k_5x_1x_2
	\end{bmatrix}
\end{equation}
and
\begin{equation}
	\Gamma(\vec{x};\vec{k})\Gamma(\vec{x};\vec{k})^T =
	\begin{bmatrix}
		k_1 + k_2x_1 + k_5x_1x_2 & k_5x_1x_2\\
		k_5x_1x_2 & k_3 + k_4x_2 + k_5x_1x_2
	\end{bmatrix}.
\end{equation}
Therefore, writing \eqref{eq:LNA} in the form \eqref{eq:stationary_A} yields
\begin{equation}\label{eq:ex4_A}
	0 = A(\vec{x},P)\vec{k} = 
	\begin{bmatrix}
		1 &         -x_1 & 0 &           0 &                                        -x_1x_2\\
		0 &           0 & 1 &         -x_2 &                                        -x_1x_2\\
		1 & x_1 - 2p_{11} & 0 &           0 &                 x_1x_2 - 2p_{12}x_1 - 2p_{11}x_2\\
		0 &       -p_{12} & 0 &       -p_{12} & x_1x_2 - p_{12}x_1 - p_{12}x_2 - p_{22}x_1 - p_{11}x_2\\
		0 &           0 & 1 & x_2 - 2p_{22} &                 x_1x_2 - 2p_{22}x_1 - 2p_{12}x_2
	\end{bmatrix}\vec{k}.
\end{equation}
Computing the reduced Gr\"{o}bner basis $\mathcal{G}$ of the ideal defined by \eqref{eq:ideal} with $A$ in \eqref{eq:ex4_A} we find that $\mathcal{G} =\{1\}$, and therefore by Theorem \ref{thm:grob_check} $\mathcal{R}_4$ is stationary globally identifiable.%
We have shown that measurements of the stationary distributions are sufficient to infer which of two biological implementations of $\mathcal{R}_4$ is better for implementing antithetic feedback control.
\end{example}

\begin{example}[Cooperative enzymatic degradation]
We now consider $\mathcal{R}_5$ shown in \eqref{eq:ex5_reactions}.
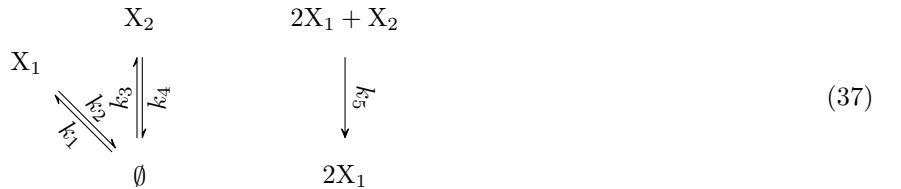
\begin{equation}\label{eq:ex5_reactions}
	\vcenter{\hbox{\schemestart
		\subscheme{$\emptyset$}
    		\arrow(z--x1){<=>[$k_2$][$k_1$]}[135]
		\subscheme{$\mathrm{X}_1$}
		\arrow(@z--x2){<=>[$k_3$][$k_4$]}[90]
		\subscheme{$\mathrm{X}_2$}
		\arrow(@x2--x12){0}[0]
		\subscheme{$2\mathrm{X}_1 + \mathrm{X}_2$}
		\arrow(@x12--x11){->[$k_5$]}[-90]
		\subscheme{$2\mathrm{X}_1$}
	\schemestop}}
\end{equation}
Note that $\mathcal{R}_5$ is similar to $\mathcal{R}_4$ considered in Example \ref{ex:seq_rate}, but the mutual degradation of $\mathrm{X}_1$ and $\mathrm{X}_2$ has been replaced by $\mathrm{X}_1$ enzymatically degrading $\mathrm{X}_2$ via the reaction \schemestart $2\mathrm{X}_1+\mathrm{X}_2$\arrow{->[$k_5$]}[0,0.5]$2X_1$\schemestop. Such an enzymatic reaction, where two copies of $\mathrm{X}_1$ bind with and degrade one copy of $\mathrm{X}_2$ is encountered when an mRNA molecule has two target sites for a complementary microRNA to bind to, both of which must be bound for degradation of the mRNA to occur~\cite{Flamand:2017vb}. For $\mathcal{R}_5$ we have that $\vec{f}(\vec{x};\vec{k})$ defined in \eqref{eq:RRE_def} is given by
\begin{equation}
	\vec{f}(\vec{x};\vec{k}) = 
	\begin{bmatrix}
		k_1 - k_2x_1 \\
	- k_5x_2x_1^2 + k_3 - k_4x2
	\end{bmatrix}
\end{equation}
and $\Gamma(\vec{x};\vec{k})$ defined in \eqref{eq:Gamma_def} is given by
\begin{equation}
	\Gamma(\vec{x};\vec{k})\Gamma(\vec{x};\vec{k})^T =
	\begin{bmatrix}
		k_1 + k_2x_1 & 0\\
		0 & k_5x_2x_1^2 + k_3 + k_4x_2
	\end{bmatrix}.
\end{equation}
Therefore, writing \eqref{eq:LNA} in the form \eqref{eq:stationary_A} yields
\begin{equation}\label{eq:ex5_A}
	0 = A(\vec{x},P)\vec{k} = 
	\begin{bmatrix}
		1 &         -x_1 & 0 &           0 &                                    0\\
		0 &           0 & 1 &         -x_2 &                             -x_1^2x_2\\
		1 & x_1 - 2p_{11} & 0 &           0 &                                    0\\
		0 &       -p_{12} & 0 &       -p_{12} &           - p_{12}x_1^2 - 2p_{11}x_2x_1\\
		0 &           0 & 1 & x_2 - 2p_{22} & x_1^2x_2 - 2p_{22}x_1^2 - 4p_{12}x_1x_2
	\end{bmatrix}\vec{k}.
\end{equation}
Computing the Gr\"{o}bner basis $\mathcal{G}$ of the ideal defined by \eqref{eq:ideal} with $A$ in \eqref{eq:ex5_A}, we find that $\mathcal{G} =\{1\}$, and therefore by Theorem \ref{thm:grob_check} $\mathcal{R}_5$ is stationary globally identifiable over $\mathbb{R}^5_{>0}$.%
\end{example}

We now apply the results of Section \ref{sec:main_result} to two different CRNs with three species.

\begin{example}[Activation cascade]\label{ex:act_cascade}
We consider a simplified model of an activation cascade $\mathcal{R}_6$, as shown in \eqref{eq:ex6_reactions_act_cascade}:
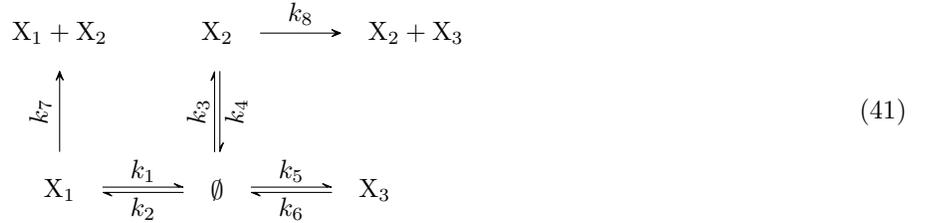
\begin{equation}\label{eq:ex6_reactions_act_cascade}
	\vcenter{\hbox{\schemestart
		\subscheme{$\emptyset$}
    		\arrow(z--x1){<=>[$k_1$][$k_2$]}[180]
		\subscheme{$\mathrm{X}_1$}
          \arrow(@z--x2){<=>[$k_3$][$k_4$]}[90]
        \subscheme{$\mathrm{X}_2$}
          \arrow(@z--x3){<=>[$k_5$][$k_6$]}[0]
        \subscheme{$\mathrm{X}_3$}
		  \arrow(@x1--x12){->[$k_7$]}[90]
        \subscheme{$\mathrm{X}_1 + \mathrm{X}_2$}
          \arrow(@x2--x23){->[$k_8$]}[0]
		\subscheme{$\mathrm{X}_2 + \mathrm{X}_3$}
	\schemestop}}
\end{equation}
In our simplified model $\mathcal{R}_6$, we have three species, $\mathrm{X}_1$, $\mathrm{X}_2$, and $\mathrm{X}_3$, each of which is a protein species. $\mathrm{X}_1$ activates the production of $\mathrm{X}_2$, which we model by the reaction \schemestart $\mathrm{X}_1$\arrow{->[$k_7$]}[0,0.5]$\mathrm{X}_1 + \mathrm{X}_2$\schemestop. Similarly, $\mathrm{X}_2$ activates the production of $\mathrm{X}_3$ as modeled by the reaction \schemestart $\mathrm{X}_2$\arrow{->[$k_7$]}[0,0.5]$\mathrm{X}_2 + \mathrm{X}_3$\schemestop. Reactions 1 through 6 model each species degrading as well as being produced at some basal rate. For $\mathcal{R}_6$, $\vec{f}(\vec{x};\vec{k})$ defined in \eqref{eq:RRE_def} is given by
\begin{equation}
	\vec{f}(\vec{x};\vec{k}) = 
	\begin{bmatrix}
		k_1 - k_2x_1\\
		k_3 - k_4x_2 + k_7x_1\\
		k_5 - k_6x_3 + k_8x_2
	\end{bmatrix}
\end{equation}
and $\Gamma(\vec{x};\vec{k})\Gamma(\vec{x};\vec{k})^T$ with $\Gamma(\vec{x};\vec{k})$ defined in \eqref{eq:Gamma_def} is given by 
\begin{equation}
	\Gamma(\vec{x};\vec{k})\Gamma(\vec{x};\vec{k})^T =
	\begin{bmatrix}
		k_1 + k_2x_1 & 0 & 0\\
		0 & k_3 + k_4x_2 + k_7x_1 & 0\\
		0 & 0 & k_5 + k_6x_3 + k_8x_2
	\end{bmatrix}.
\end{equation}
Therefore, writing \eqref{eq:LNA} in the form \eqref{eq:stationary_A} yields
\begin{equation}\label{eq:ex6_A}
	0 = A(\vec{x},P)\vec{k} = 
	\begin{bmatrix}
		1 &         -x_1 & 0 &           0 & 0 &           0 &           0 &           0\\
		0 &           0 & 1 &         -x_2 & 0 &           0 &          x_1 &           0\\
		0 &           0 & 0 &           0 & 1 &         -x_3 &           0 &          x_2\\
		1 & x_1 - 2p_{11} & 0 &           0 & 0 &           0 &           0 &           0\\
		0 &       -p_{12} & 0 &       -p_{12} & 0 &           0 &        p_{11} &           0\\
		0 &       -p_{13} & 0 &           0 & 0 &       -p_{13} &           0 &        p_{12}\\
		0 &           0 & 1 & x_2 - 2p_{22} & 0 &           0 & 2p_{12} + x_1 &           0\\
		0 &           0 & 0 &       -p_{23} & 0 &       -p_{23} &        p_{13} &        p_{22}\\
		0 &           0 & 0 &           0 & 1 & x_3 - 2p_{33} &           0 & 2p_{23} + x_2
	\end{bmatrix}\vec{k}.
\end{equation}
Computing the reduced Gr\"{o}bner basis $\mathcal{G}$ of the ideal \eqref{eq:ideal} with $A$ given in \eqref{eq:ex6_A}, we find that $\mathcal{G} = \{1\}$, and therefore by Theorem \ref{thm:grob_check} $\mathcal{R}_6$ is stationary globally identifiable over $\mathbb{R}^8_{>0}$.%
\end{example}

\begin{example}[Coupled sequestration reactions]
We now consider a biological system with three species $\mathrm{X}_1$, $\mathrm{X}_2$, and $\mathrm{X}_3$ where $\mathrm{X}_2$ binds to and mutually degrades with both $\mathrm{X}_1$ and $\mathrm{X}_3$. We model this system by the CRN shown in \eqref{eq:ex7_reactions}:
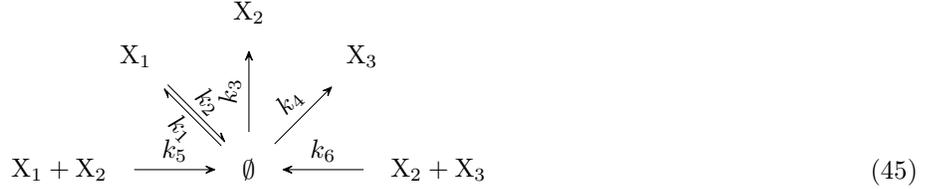
\begin{equation}\label{eq:ex7_reactions}
	\schemestart
		\subscheme{$\emptyset$}
    		\arrow(z--x1){<=>[$k_2$][$k_1$]}[135]
		\subscheme{$\mathrm{X}_1$}
		\arrow(@z--x2){->[$k_3$]}[90]
		\subscheme{$\mathrm{X}_2$}
		\arrow(@z--x3){->[$k_4$]}[45]
		\subscheme{$\mathrm{X}_3$}
		\arrow(@z--x12){<-[$k_5$]}[180]
        \subscheme{$\mathrm{X}_1 + \mathrm{X}_2$}
		\arrow(@z--x23){<-[$k_6$]}[0]
		\subscheme{$\mathrm{X}_2 + \mathrm{X}_3$}
	\schemestop
\end{equation}
We assume that all three species are produced at some rate, but only $\mathrm{X}_1$ spontaneously degrades. This CRN is a coarse model of two RNA species ($\mathrm{X}_1$ and $\mathrm{X}_3$), which are degraded by the same microRNA species ($\mathrm{X}_2$). Such systems are common in biology, as some microRNA species are known to regulate multiple genes by targeting the corresponding mRNA species~\cite{linsley:2007wg}. For $\mathcal{R}_7$ the definition of $\vec{f}(\vec{x};\vec{k})$ in \eqref{eq:RRE_def} gives
\begin{equation}
	\vec{f}(\vec{x};\vec{k}) = 
	\begin{bmatrix}
		k_1 - k_2x_1 - k_5x_1x_2\\
		k_3 - k_5x_1x_2 - k_6x_2x_3\\
		k_4 - k_6x_2x_3
	\end{bmatrix}
\end{equation}
and using the definition of $\Gamma(\vec{x};\vec{k})$ given in \eqref{eq:Gamma_def} we obtain that
\begin{equation}
	\Gamma(\vec{x};\vec{k})\Gamma(\vec{x};\vec{k})^T =
	\begin{bmatrix}
		k_1 + k_2x_1 + k_5x_1x_2 & k_5x_1x_2 & 0\\
		k_5x_1x_2 & k_3 + k_5x_1x_2 + k_6x_2x_3 & k_6x_2x_3\\
		0 & k_6x_2x_3 & k_4 + k_6x_2x_3
	\end{bmatrix}.
\end{equation}
Therefore, writing \eqref{eq:LNA} in the form \eqref{eq:stationary_A} yields $0 = A(\vec{x},P)\vec{k}$, where  
\begin{equation}\label{eq:ex7_A}
\begin{multlined}
	A(\vec{x},P) = \\
	\mbox{\footnotesize
	$
	\begin{bmatrix}
		1 &         -x_1 & 0 & 0 &                                        -x_1x_2 &                                             0\\
		0 &           0 & 1 & 0 &                                        -x_1x_2 &                                        -x_2x_3\\
		0 &           0 & 0 & 1 &                                             0 &                                        -x_2x_3\\
		1 & x_1 - 2p_{11} & 0 & 0 &                 x_1x_2 - 2p_{12}x_1 - 2p_{11}x_2 &                                             0\\
		0 &       -p_{12} & 0 & 0 & x_1x_2 - p_{12}x_1 - p_{12}x_2 - p_{22}x_1 - p_{11}x_2 &                           - p_{12}x_3 - p_{13}x_2\\
		0 &       -p_{13} & 0 & 0 &                           - p_{13}x_2 - p_{23}x_1 &                           - p_{12}x_3 - p_{13}x_2\\
		0 &           0 & 1 & 0 &                 x_1x_2 - 2p_{22}x_1 - 2p_{12}x_2 &                 x_2x_3 - 2p_{23}x_2 - 2p_{22}x_3\\
		0 &           0 & 0 & 0 &                           - p_{13}x_2 - p_{23}x_1 & x_2x_3 - p_{23}x_2 - p_{23}x_3 - p_{33}x_2 - p_{22}x_3\\
		0 &           0 & 0 & 1 &                                             0 &                 x_2x_3 - 2p_{33}x_2 - 2p_{23}x_3
	\end{bmatrix}$}.
\end{multlined}
\end{equation}
Computing the reduced Gr\"{o}bner basis $\mathcal{G}$ of the ideal \eqref{eq:ideal} with $A$ given in \eqref{eq:ex7_A}, we find that $\mathcal{G} = \{1\}$, and therefore by Theorem \ref{thm:grob_check} $\mathcal{R}_7$ is stationary globally identifiable over $\mathbb{R}^6_{>0}$.%
\end{example}
An example of a non-identifiable CRN is provided in Example \ref{ex:feedback}, which is deferred until Section \ref{sec:idble:extrinsic}.
\section{Model discrimination}\label{sec:idble_model_id}
An application of the results of  Section \ref{sec:main_result} is to certifying a type of identifiability where instead of asking if it is possible to uniquely determine the value of $\vec{k}$, we ask if it is possible to determine whether the rate constant vector $\vec{k}$ is in $K_1 \subseteq \mathbb{R}_{\geq 0}^r$ or in $K_2 \subseteq \mathbb{R}_{\geq 0}^r$. For example, we may be interested in determining which of two reactions is present in our system, with the knowledge that at most one of the two reactions is present. This notion is formalized in the following definition.
\begin{definition}
	A CRN $\mathcal{R}$ is stationary model discriminable between $K_1$ and $K_2$ if there does not exist $\vec{k}_1\in K_1, \vec{k}_2\in K_2$ such that $\mathcal{R}(\vec{k}_1) = \mathcal{R}(\vec{k}_2)$. 
\end{definition}
In this work, we do not give a complete characterization of stationary model discriminability in our problem setting, however we do present the following result, which allows us to directly apply the framework  developed in this work to certify stationary model discriminability for CRNs.
We first consider how to certify that a CRN is stationary globally identifiable over a general set $K$ defined in terms of polynomial equations. To do this, we consider a set
\begin{equation}\label{eq:lifted_K}
		\bar{K} = \left\{ (\vec{k},\vec{y}) \in \mathbb{R}^{r+l} \middle| h_i(\vec{k},\vec{y}) = 0, \;i = 1,2,\dots,p \right\}
\end{equation}
where $h_i(\vec{k},\vec{y})$ are polynomials such that the orthogonal projection of $\bar{K}$ onto the $\vec{k}$ space is equal to $K$. We call such a $\bar{K}$ a lifted representation of $K$. If $K$ is a semialgebraic set, that is, a finite union of sets described by polynomial equalities and inequalities, then it is always possible to construct a lifted representation as in \eqref{eq:lifted_K} with $l=1$ \cite{motzkin1970real}. A simple way to convert a strict inequality of the form $p(\vec{x}) >  0$, to an equality is by adding a variable $y$, and using the constraint $p(\vec{x})y^2 - 1 = 0$. Similarly, an inequality of the form $p(\vec{x}) \geq 0$ can be converted to an equality by adding a variable $y$ and using the constraint $p(\vec{x}) - y^2 = 0$~\cite{bochnak2013real,bertsekas1997nonlinear}.
\begin{theorem}\label{thm:general_K}
	Consider a CRN $\mathcal{R}$ and a set $K$ such that $\bar{K}$ defined in \eqref{eq:lifted_K} is a lifted representation of $K$. If the reduced Gr\"{o}bner basis of
	\begin{equation}\label{eq:general_K_ideal}
		\left\langle h_j(\vec{k},\vec{y}),\; j=1,\dots,p, \;A_q(\vec{x},P)\vec{k}\; q=1,\dots,r,\;M^{(r-1)\times (r-1)}_i(\vec{x},P)\; i=1,\dots,m \right\rangle
	\end{equation}
	is $\left\{1\right\}$, then $\mathcal{R}$ is stationary globally identifiable over $K$.
\end{theorem}
\begin{proof}
	The proof follows that of Theorem \ref{thm:grob_check}, however we replace the polynomials $k_iy_i^2 - 1$ with $h_j(\vec{k},\vec{y})$, and instead of Theorem \ref{thm:spi_alg} we have only a sufficient semialgebraic condition for stationary global identifiability, since here we do not assume that $K$ is open. Suppose $\mathcal{R}$ is not stationary globally identifiable over $K$. Then there exist $\vec{k}_1,\vec{k}_2 \in K$, $\vec{x}_0 \in \mathbb{R}_{>0}^n$, and $P \in \mathbb{S}^{n\times n}$ such that $\vec{k}_1$ and $\vec{k}_2$ are linearly independent and $(\vec{x}_0,P_0) = \mathcal{R}(\vec{k}_1) = \mathcal{R}(\vec{k}_2)$. The fact that $\vec{k}_1 \in K$ implies that there exists $\vec{y}_1$ such that $(\vec{k}_1,\vec{y}_1)\in\bar{K}$. By the fact that $\vec{k}_1$ and $\vec{k}_2$ are linearly independent, $\rank A(\vec{x}_0,P_0) < r-1$, and hence $M^{(r-1)\times (r-1)}_i(\vec{x}_0,P_0)$ for all $i = 1,\dots, m$. Since additionally $0 = A_q(\vec{x}_0,P_0)\vec{k}_1$, we have that $\vec{k} = \vec{k}_1$, $\vec{y} = \vec{y}_1$, $\vec{x} = \vec{x}_0$, $P = P_0$ is a solution to 
	\begin{align}
		0=& h_j(\vec{k},\vec{y}),\;\forall j=1,\dots,p\\
		 0 =& A_q(\vec{x},P)\vec{k},\; \forall q=1,\dots,r\\
		 0  =& M^{(r-1)\times (r-1)}_i(\vec{x},P),\;\forall i = 1,\dots,m.
	\end{align}
	Therefore, by Theorem \ref{thm:nullstellensatz_coro}, the reduced Gr\"{o}bner basis of \eqref{eq:general_K_ideal} must not be $\left\{1\right\}$. We have thus shown the contrapositive of the theorem statement.
\end{proof}

\setcounter{example}{0}
\begin{example}[Example 1 with a different set $K$]
	We return to Example \ref{ex:one_dim}, however instead of asking if $\mathcal{R}_1$ given by \eqref{eq:ex1_reactions} is stationary globally identifiable over $\mathbb{R}_{>0}^3$, we are interested in investigating whether it is stationary globally identifiable over 
	\begin{equation}\label{eq:ex_general_K_K}
		K = \left\{ \vec{k}\in \mathbb{R}^3 \middle| k_1 > 0,\; k_2 >0,\; k_3 \geq 0\right\}.
	\end{equation}
	One way to represent this set as the projection of a set $\bar{K}$ in the form  \eqref{eq:lifted_K} is by choosing $\bar{K}$ as:
	\begin{equation}\label{eq:ex_general_K_K_bar}
		\bar{K} = \left\{ (\vec{k},\vec{y})\in\mathbb{R}^6\middle| y_1^2k_1 - 1 = 0,\; y_2^2k_2 - 1 = 0,\; k_3 - y_3^2 = 0\right\}.
	\end{equation}
	Indeed, it can be checked that the orthogonal projection of $\bar{K}$ onto $\vec{x}$ is $K$. In fact, if $y_i^2k_i - 1 = 0$ then $k_1 = 1/y_i^2 > 0$. Similarly, if $k_2 - y_2^2 = 0$, then $k_2 = y_2^2 \geq 0$.
	To apply Theorem \ref{thm:general_K} we must compute the reduced Gr\"{o}bner basis of \eqref{eq:general_K_ideal}, which from \eqref{eq:ex1_A} is given by
	\begin{equation}\label{eq:ex_general_K_ideal}
	\begin{multlined}
		\left\langle k_1 - k_2x_1 - k_3x_1^2,\;
	k_1 - k_3(4p_{11}x_1 - x_1^2) - k_2(2p_{11} - x_1),\;
	 k_1y_1^2 - 1,\right.\\\left. k_2y_2^2 - 1,\; k_3 - y_3^2,\;
	 2x_1 - 2p_{11},\; 2x_1^2 - 4p_{11}x_1,\; 2p_{11}x_1^2\right\rangle.
	\end{multlined}
	\end{equation}
	Using Macaulay2 \cite{M2} we find that the reduced Gr\"{o}bner basis of \eqref{eq:ex_general_K_ideal} is $\{1\}$, and hence by Theorem \ref{thm:general_K} $\mathcal{R}_1$ is stationary globally identifiable over $K$ given by \eqref{eq:ex_general_K_K}.
\end{example}

We are now ready to study the model discriminability problem. Our approach is to attempt to certify global stationary identifiability of $\mathcal{R}$ over the set $K_1 \cup K_2$, which is formalized in the following theorem.

\begin{theorem}\label{thm:spi_implies_model_idble}
	Consider a CRN $\mathcal{R}$. Let $K_1,K_2 \subset \mathbb{R}^r_{\geq 0}$ be such that $\mathrm{cone}\left(K_1\right) \cap K_2 = \emptyset$\footnote{For a set $K\subseteq \mathbb{R}^v$, $\mathrm{cone}(K) = \left\{ \vec{z}\in\mathbb{R}^v\middle| \vec{z} = \lambda \vec{k},\; \vec{k}\in K, \; \lambda \geq 0 \right\}$. }. If $\mathcal{R}$ is stationary globally identifiable over $K = K_1 \cup K_2$, then $\mathcal{R}$ is stationary model discriminable between $K_1$ and $K_2$.
\end{theorem}
\begin{proof}
	We prove Theorem \ref{thm:spi_implies_model_idble} by contraposition. Suppose that $\mathcal{R}$ is not stationary model discriminable between $K_1$ and $K_2$. Then there exists $\vec{k}_1 \in K_1$ and $\vec{k}_2 \in K_2$ such that $\mathcal{R}(\vec{k}_1) = \mathcal{R}(\vec{k}_2)$. The assumption that $\vspan\left(K_1\right) \cap K_2 = \emptyset$ ensures that there does not exist $\alpha$ such that $\vec{k}_1 = \alpha \vec{k}_2$, and hence $\mathcal{R}$ is not stationary globally identifiable over $K_1 \cup K_2$.
\end{proof}

\begin{remark}
	The converse of Theorem \ref{thm:spi_implies_model_idble} is not true. However, Theorem \ref{thm:spi_implies_model_idble} provides a sufficient condition to conclude that $\mathcal{R}$ is stationary model identifiable between $K_1$ and $K_2$.
\end{remark}

As an illustration, suppose that for some CRN $\mathcal{R}$ with $r$ reactions, we know that exactly one between the $r$\textsuperscript{th} and $r-1$\textsuperscript{th} reactions is present. If we want to determine if it is possible to discriminate from the stationary distribution of $\mathcal{R}$ between reaction $r$ being present and reaction $r-1$ being present, we ask if $\mathcal{R}$ is stationary model discriminable between $K_1$ and $K_2$ where, letting $\vec{k}_{1:r-2}$ be the vector of the first $r-2$ elements of $\vec{k}$,
\begin{equation}
	K_1  = \left\{ \vec{k} \in \mathbb{R}_{\geq 0}^r \middle| \vec{k}_{1:r-2}>0,\; k_{r-1} > 0\;\mbox{and}\; k_r = 0\right\}
\end{equation}
and
\begin{equation}
	K_2  = \left\{ \vec{k} \in \mathbb{R}_{\geq 0}^r \middle| \vec{k}_{1:r-2}>0,\; k_{r-1} = 0\;\mbox{and}\; k_r > 0\right\}.
\end{equation}
Let $K = K_1 \cup K_2$. We need a representation of $K$ as in equation \eqref{eq:lifted_K}. One such representation of $K$ is
\begin{equation}\label{eq:lifted_model_selection_K_bar}
\begin{multlined}
	\bar{K} = \left\{(\vec{k},\vec{y})\in\mathbb{R}^{2r+1} \middle| 0 = k_iy_i^2 - 1, \;i=1,2,\dots,r-2, 0  = k_{r-1} - y_{r-1}^2,\right. \\
	 \left. 0 = k_r-y_r^2,\; 0 = k_{r-1}k_r,\; 0 = (k_{r-1}+k_r)y_{r+1}^2 - 1\right\}.
\end{multlined}
\end{equation}
\begin{remark}
We can choose $\bar{K}$ to be any lifted representation of $K_1 \cup K_2$ of the form \eqref{eq:lifted_K}, however, it is possible for the reduced Gr\"{o}bner basis of \eqref{eq:general_K_ideal} to be $\{1\}$ for some choices of $\bar{K}$ and not $\{1\}$ for other choices of $\bar{K}$. Such a possibility is a consequence of using Nullstellensatz to prove identifiability, and using Positivstellensatz as discussed in Remark \ref{remark:positivstellensatz} would prevent this issue.
\end{remark}
\setcounter{example}{0}
\begin{example}[1-dimensional model discriminability]
	Let us again consider $\mathcal{R}_1$ given by \eqref{eq:ex1_reactions}. Suppose we know that either $k_2>0$ and $k_3 = 0$, or $k_2 = 0$ and $k_3>0$. If we are interested in whether we can discriminate between these two models, we use the framework of this section as follows. Let
	\begin{equation}\label{eq:ex_1d_discr_K1}
		K_1 = \left\{ \vec{k}\in\mathbb{R}^3\middle| k_1>0,\; k_2 >0,\; k_3 = 0\right\}
	\end{equation}
	and
	\begin{equation}\label{eq:ex_1d_discr_K2}
		K_2 = \left\{\vec{k}\in\mathbb{R}^3\middle| k_1>0,\;k_2 = 0; k_3>0\right\}.
	\end{equation}
	Then, to check if $\mathcal{R}_1$ is stationary model discriminable between $K_1$ and $K_2$ we let $K = K_1 \cup K_2$, which has lifted representation
	\begin{equation}
		\bar{K} = \left\{(\vec{k},\vec{y})\in\mathbb{R}^{7} \middle| 0 = k_1y_1^2 -1,\; 0 = k_2 - y_2^2,\; 0 = k_3 - y_3^2,\; 0 = k_2k_3,\; 0 = (k_2+k_3)y_4^2 - 1 \right\}.
	\end{equation}
	In this case, using \eqref{eq:ex1_A} and $h_j(\vec{k},\vec{y})$ defined in \eqref{eq:lifted_model_selection_K_bar}, the ideal given by \eqref{eq:general_K_ideal} is
	\begin{equation}\label{eq:ex_1d_discr_ideal}
	\begin{multlined}
		\left\langle k_1 - k_2x_1 - k_3x_1^2,\;
	k_1 - k_3(4p_{11}x_1 - x_1^2) - k_2(2p_{11} - x_1),\;\right.\\
	 \left. k_1y_1^2 - 1, k_2 - y_2^2, k_3 - y_3^2,  k_2k_3, (k_2+k_3)y_4^2 - 1\;
	 2x_1 - 2p_{11}, 2x_1^2 - 4p_{11}x_1, 2p_{11}x_1^2\right\rangle.
	\end{multlined}
	\end{equation}
	Using Macaulay2 \cite{M2}, we find that the reduced Gr\"{o}bner basis of \eqref{eq:ex_1d_discr_ideal} is $\{1\}$, and hence by Theorems \ref{thm:spi_implies_model_idble} and \ref{thm:general_K} the CRN $\mathcal{R}_1$ is stationary model discriminable between $K_1$ and $K_2$ given by \eqref{eq:ex_1d_discr_K1} and \eqref{eq:ex_1d_discr_K2}, respectively.
\end{example}

\subsection{Examples}
We now use \eqref{eq:lifted_model_selection_K_bar} to certify stationary model discriminability of several biologically relevant systems via Theorem \ref{thm:spi_implies_model_idble}.

\setcounter{example}{7}
\begin{example}[Determining the direction of an activation (model discrimination)]\label{ex:which_act}
In this example we consider whether it is possible to determine from only measurements of the joint stationary distribution of two genes $\mathrm{X}_1$ and $\mathrm{X}_2$ whether $\mathrm{X}_1$ activates $\mathrm{X}_2$ or $\mathrm{X}_2$ activates $\mathrm{X}_1$. Such a question is of practical importance in systems biology because it asks whether one can deduce causality in a biological system without observing how the system evolves over time, or how it reacts to applied perturbations. This question is conceptually related to the study of causal inference, though here we ask whether we can distinguish between two \emph{a prior} given stochastic process models, instead of deciding between graphical models~\cite{peters2017elements}. Such a system is conceptually modeled by CRN $\mathcal{R}_8$ shown in \eqref{eq:ex8_reactions}.
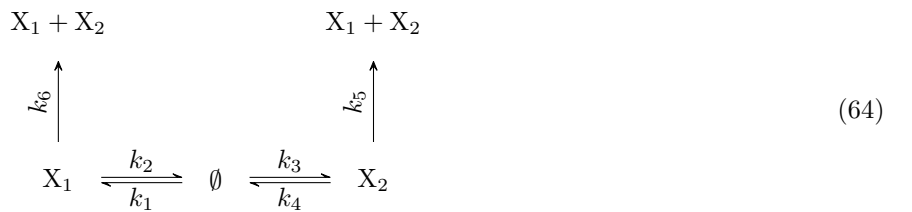
\begin{equation}\label{eq:ex8_reactions}
	\vcenter{\hbox{\schemestart
		\subscheme{$\emptyset$}
    		\arrow(z--x1){<=>[$k_2$][$k_1$]}[180]
		\subscheme{$\mathrm{X}_1$}
          \arrow(@z--x2){<=>[$k_3$][$k_4$]}[0]
        \subscheme{$\mathrm{X}_2$}
		  \arrow(@x1--x12){->[$k_6$]}[90]
        \subscheme{$\mathrm{X}_1 + \mathrm{X}_2$}
          \arrow(@x2--x23){->[$k_5$]}[90]
		\subscheme{$\mathrm{X}_1 + \mathrm{X}_2$}
	\schemestop}}
\end{equation}
We note that in order to simplify the system we have modeled gene expression as a one step process, and model activation of $\mathrm{X}_2$ by $\mathrm{X}_1$ with the reactions \schemestart $\emptyset$\arrow{->[$k_3$]}[0,0.5]$\mathrm{X}_2$\schemestop \; and  \schemestart $\mathrm{X}_1$\arrow{->[$k_6$]}[0,0.5]$X_1 + X_2$\schemestop, i.e., an affine activation function of the form $k_3 + k_6x_1$. The activation of $\mathrm{X}_1$ by $\mathrm{X}_2$ is modeled analogously via the 1\textsuperscript{st} and 5\textsuperscript{th} reactions. For $\mathcal{R}_8$ $\vec{f}(\vec{x};\vec{k})$ defined in \eqref{eq:RRE_def} is given by
\begin{equation}
	\vec{f}(\vec{x};\vec{k}) = 
	\begin{bmatrix}
		k_1 - k_2x_1 + k_5x_2 \\
		k_3 - k_4x_2 + k_6x_1
	\end{bmatrix}
\end{equation}
and $\Gamma(\vec{x};\vec{k})$ as defined in \eqref{eq:Gamma_def} is given by
\begin{equation}
	\Gamma(\vec{x};\vec{k})\Gamma(\vec{x};\vec{k})^T =
	\begin{bmatrix}
		k_1 + k_2x_1 + k_5x_2 & 0\\
		0 & k_3 + k_4x_2 + k_6x_1
	\end{bmatrix}.
\end{equation}
Therefore, writing \eqref{eq:LNA} in the form \eqref{eq:stationary_A} yields
\begin{equation}\label{eq:ex8_A}
	0 = A(\vec{x},P)\vec{k} = 
	\begin{bmatrix}
		1 &         -x_1 & 0 &           0 &           0 &          x_2\\
		0 &           0 & 1 &         -x_2 &          x_1 &           0\\
		1 & x_1 - 2p_{11} & 0 &           0 &           0 & 2p_{12} + x_2\\
		0 &       -p_{12} & 0 &       -p_{12} &        p_{11} &        p_{22}\\
		0 &           0 & 1 & x_2 - 2p_{22} & 2p_{12} + x_1 &           0
	\end{bmatrix}\vec{k}.
\end{equation}
The two models we wish to decide between are
\begin{enumerate}
	\item $\mathrm{X}_1$ is constitutively expressed ($k_1 >0$) and activates $\mathrm{X}_2$ ($k_3,k_6 > 0$),
	\item $\mathrm{X}_2$ is constitutively expressed ($k_3 > 0$) and activates $\mathrm{X}_1$ ($k_1,k_5 > 0$).
\end{enumerate}
In both models we assume $\mathrm{X}_1$ and $\mathrm{X}_2$ degrade at a nonzero rate ($k_2,k_4 >0$). Using the framework of Section \ref{sec:idble_model_id} we represent model 1 as the reaction rate vector being in
\begin{equation}\label{eq:K1_6species}
	K_1  = \left\{ \vec{k} \in \mathbb{R}_{\geq 0}^6 \middle| \vec{k}_{1:4}>0,\; k_{5} > 0\;\mbox{and}\; k_6 = 0\right\}
\end{equation}
and model 2 by the reaction rate vector being in 
\begin{equation}\label{eq:K2_6species}
	K_2  = \left\{ \vec{k} \in \mathbb{R}_{\geq 0}^6 \middle| \vec{k}_{1:4}>0,\; k_{5} = 0\;\mbox{and}\; k_6 > 0\right\}.
\end{equation}
In this case \eqref{eq:lifted_model_selection_K_bar} becomes
\begin{equation}\label{eq:lifted_model_selection_K_bar_6species}
\begin{multlined}
	\bar{K} = \left\{(\vec{k},\vec{y})\in\mathbb{R}^{2r+1} \middle|
	0 = k_iy_i^2 - 1, \;i=1,2,\dots,4, \right.\\\left.
	0  = k_{5} - y_{5}^2,\; 0 = k_6-y_6^2,\;
	0 = k_{5}k_6,\;
	0 = (k_{5}+k_6)y_{7}^2 - 1\right\},
\end{multlined}
\end{equation}
which we use as our representation of $K = K_1 \cup K_2$.  Computing the Gr\"{o}bner basis $\mathcal{G}$ of the ideal defined by \eqref{eq:ideal} with $A$ in \eqref{eq:ex8_A}, we find that $\mathcal{G} =\{1\}$, and therefore by Theorem \ref{thm:grob_check} $\mathcal{R}_8$ is stationary globally identifiable over $K_1 \cup K_2$.%
We can therefore conclude by Theorem \ref{thm:spi_implies_model_idble} that $\mathcal{R}_8$ is stationary model discriminable between $K_1$ and $K_2$. This result conflicts with the intuition that correlation between the concentrations of $\mathrm{X}_1$ and $\mathrm{X}_2$ is insufficient to infer whether $\mathrm{X}_1$ ``causes'' $\mathrm{X}_2$ or \emph{vice versa}. However, examining the joint distribution allows us to tell which direction the activation acts because the noise on $x_1$ will contribute to the variance of $x_2$ when $\mathrm{X}_1$ activates $\mathrm{X}_2$, whereas the noise on $x_2$ will contribute to the variance of $x_1$ when $\mathrm{X}_2$ activates $\mathrm{X}_1$. The fact that noise from ``upstream'' genes contributes to a higher variance in ``downstream'' genes is well understood \cite{Ochab-Marcinek:2010aa}, though to the authors' knowledge the use of this principle for model discrimination has not been explored.
\end{example}
\begin{remark}\label{rem:discrimination}
	In Example \ref{ex:which_act} we showed that in CRN $\mathcal{R}_8$ it is possible to determine whether reaction 5 or 6 is present. Given sufficient data, the inference can be carried out by solving
	\begin{equation}
		c_1 = \min_{\vec{k} \in K_1} \norm{A(\hat{\vec{x}},\hat{P})\vec{k}}_2^2
	\end{equation}
	and
	\begin{equation}
		c_2 = \min_{\vec{k} \in K_2} \norm{A(\hat{\vec{x}},\hat{P})\vec{k}}_2^2,
	\end{equation}
	where $\hat{\vec{x}}$ is the sample mean and $\hat{P}$ is $\Omega$ times the sample covariance. This procedure is very similar to standard model selection methods \cite{akaike1973information}, expect that the fitting of the parameters is not done via maximum likelihood estimation, and we do not worry about the Occam factor present in the Akaike information criterion, since given infinite data, exactly one of $c_1$ and $c_2$ will be zero. In this case, if $c_1 = 0$ then $\mathrm{X}_1$ is constitutively expressed ($k_1 >0$) and activates $\mathrm{X}_2$ ($k_3,k_6 > 0$), whereas if $c_2 = 0$ then $\mathrm{X}_2$ is constitutively expressed ($k_3 > 0$) and activates $\mathrm{X}_1$ ($k_1,k_5 > 0$).
\end{remark}

\begin{example}[Sequestration vs enzymatic degradation]\label{ex:seq_vs_enzy}
As discussed in Example \ref{ex:seq_rate}, the antithetic motif where $\mathrm{X}_1$ and $\mathrm{X}_2$ mutually degrade is important to constructing integral biomolecular feedback controllers. When searching for pairs of species that can be used to implement such a controller, it is common that it is not know \emph{a priori} whether $\mathrm{X}_1$ and $\mathrm{X}_2$ mutually degrade, or whether one enzymatically degrades the other. Since integral controllers using an antithetic motif are designed assuming that $\mathrm{X}_1$ and $\mathrm{X}_2$ mutually degrade, it is important to be able to distinguish between these two models~\cite{qian2018realizing,briat2016antithetic}. Typically, detailed kinetic studies need to be done to determine which model is accurate for the interaction between two given species~\cite{zheng2015analytical}. Here, we investigate if an alternative experimental approach where only the stationary distribution of a system of $\mathrm{X}_1$ and $\mathrm{X}_2$ is measured can be used to answer this model discrimination question. Consider the CRN $\mathcal{R}_9$ shown in \eqref{eq:ex9_reactions}:
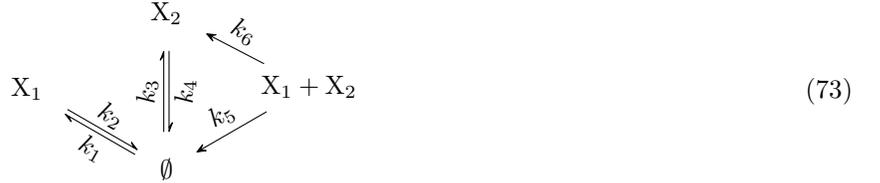
\begin{equation}\label{eq:ex9_reactions}
	\vcenter{\hbox{\schemestart
		\subscheme{$\emptyset$}
    		\arrow(z--x1){<=>[$k_2$][$k_1$]}[150]
		\subscheme{$\mathrm{X}_1$}
		\arrow(@z--x2){<=>[$k_3$][$k_4$]}[90]
		\subscheme{$\mathrm{X}_2$}
		\arrow(@z--x12){<-[$k_5$]}[30]
		\subscheme{$\mathrm{X}_1 + \mathrm{X}_2$}
        \arrow(@x2--@x12){<-[$k_6$]}
	\schemestop}}
\end{equation}
For $\mathcal{R}_9$ we have from \eqref{eq:RRE_def} that
\begin{equation}
	\vec{f}(\vec{x};\vec{k}) = 
	\begin{bmatrix}
		k_1 - k_2x_1 - k_5x_1x_2 - k_6x_1x_2 \\
		k_3 - k_4x_2 - k_5x_1x_2
	\end{bmatrix},
\end{equation}
and from \eqref{eq:Gamma_def} that
\begin{equation}
	\Gamma(\vec{x};\vec{k})\Gamma(\vec{x};\vec{k})^T =
	\begin{bmatrix}
		k_1 + k_2x_1 + k_5x_1x_2 + k_6x_1x_2 & k_5x_1x_2\\
		k_5x_1x_2 & k_3 + k_4x_2 + k_5x_1x_2
	\end{bmatrix}.
\end{equation}
Therefore, writing \eqref{eq:LNA} in the form \eqref{eq:stationary_A} yields $0 = A(\vec{x},P)\vec{k}$ where
\begin{equation}\label{eq:ex9_A}
\begin{multlined}
	 A(\vec{x},P) = \\
	 \mbox{\small$
	\begin{bmatrix}
		1 &         -x_1 & 0 &           0 &                                        -x_1x_2 &                        -x_1x_2\\
		0 &           0 & 1 &         -x_2 &                                        -x_1x_2 &                             0\\
		1 & x_1 - 2p_{11} & 0 &           0 &                 x_1x_2 - 2p_{12}x_1 - 2p_{11}x_2 & x_1x_2 - 2p_{12}x_1 - 2p_{11}x_2\\
		0 &       -p_{12} & 0 &       -p_{12} & x_1x_2 - p_{12}x_1 - p_{12}x_2 - p_{22}x_1 - p_{11}x_2 &           - p_{12}x_2 - p_{22}x_1\\
		0 &           0 & 1 & x_2 - 2p_{22} &                 x_1x_2 - 2p_{22}x_1 - 2p_{12}x_2 &                             0
	\end{bmatrix}$}.
\end{multlined}
\end{equation}
Here we consider the additional assumption that exactly one of the two degradation reactions involving $\mathrm{X}_1$ and $\mathrm{X}_2$ is present with a nonzero rate. Asking if we can discriminate between these two cases is asking if $\mathcal{R}_9$ is model discriminable between
\begin{enumerate}
	\item $\mathrm{X}_1$ and $\mathrm{X}_2$ mutually degrade ($k_5 > 0$),
	\item $\mathrm{X}_2$ enzymatically degrades $\mathrm{X}_1$ ($k_6 >0$).
\end{enumerate}
In both models we assume $\mathrm{X}_1$ and $\mathrm{X}_2$ are constitutively produced ($k_1,k_3 >0$) and dilute/spontaneously degrade ($k_2,k_4 >0$). The model discrimination problem is then as in Example \ref{ex:which_act} between $\vec{k}$ being in $K_1$ given by \eqref{eq:K1_6species} and $K_2$ given by \eqref{eq:K2_6species}. As in Example \ref{ex:which_act}, we construct a lifted representation of $K = K_1 \cup K_2$ as \eqref{eq:lifted_model_selection_K_bar_6species}. We perform the same procedure as in Example \ref{ex:which_act}, computing the Gr\"{o}bner basis $\mathcal{G}$ of the ideal \eqref{eq:ideal} with $A$ given in \eqref{eq:ex6_A}.%
In this case we find that $\mathcal{G} =\{1\}$, and therefore by Theorem \ref{thm:grob_check} $\mathcal{R}_9$ is stationary globally identifiable over $K_1 \cup K_2$. We therefore conclude by Theorem \ref{thm:spi_implies_model_idble} that $\mathcal{R}_9$ is stationary model discriminable between $K_1$ and $K_2$.
\end{example}
\begin{remark}
	Given data drawn from the stationary distribution of $x_1$ and $x_2$ in $\mathcal{R}_9$, the same technique described in Remark \ref{rem:discrimination} can be used to determine which model for the interaction of $\mathrm{X}_1$ and $\mathrm{X}_2$ is present in the system.
\end{remark}

\section{Gaining identifiability with extrinsic noise}\label{sec:idble:extrinsic}
We now extend our methods to handle CRNs with extrinsic noise. Our motivation is models of genetic circuits on plasmids, where the plasmid copy number, and therefore certain reaction rate constants in the CRN, vary among cells in the population~\cite{DelVecchio:2015book}. To this end, we consider systems where this variation across cells, or \emph{extrinsic} noise, denoted by $\vec{u} = [u_1,u_2,\dots,u_s]^T$, is an element of the set $U\subset \mathbb{R}^s$, with known distribution $\rho(\vec{u})$, and the reaction rate constants are given by $\vec{g}(\vec{u}^i)\odot \vec{k}$, where $\vec{k}$ is the nominal reaction rate constants and $\vec{g}:U\rightarrow\mathbb{R}_{\geq 0}^r$ is a known  function representing how $\vec{u}\in U$ perturbs $\vec{k}$. Here ``$\odot$'' denotes elementwise multiplication. Our assumption that $\vec{g}(\vec{u})$ is known requires a mechanistic model of how the extrinsic noise enters the system. For simplicity, in this work we assume $|U|<\infty$ as well as that within each cell the value of $\vec{u}$ is constant. In this case, the population distribution after all cells have reached their stationary distribution is given by a Gaussian mixture model of the form
\begin{equation}\label{eq:extrinsic_noise_stationary_dist}
	f_{X}(\vec{x};\vec{k}) = \sum_{\vec{u} \in U} \rho(\vec{u}) v(\vec{x};\mathcal{R}\left(\vec{g}(\vec{u})\odot\vec{k}\right))
\end{equation}
where $v(\vec{x};R)$ denotes the Gaussian probability density function with parameters $R = (\vec{x}',P')$, where the mean is $\vec{x}'$ and the covariance is $P'$.

\begin{example}[1-dimensional extrinsic noise]\label{ex:1d_extrinsic}
	We consider a variation on $\mathcal{R}_1$, where extrinsic noise affects the rate of reaction 1. This corresponds to a system where $\mathrm{X}_1$ is a protein species produced at a rate proportional to the DNA copy number in a given cell~\cite{DelVecchio:2015book}. For simplicity, we assume that in each cell there is either zero copies, one copy, or two copies of the gene coding for $\mathrm{X}$, with probability $1/2$, $1/4$, and $1/4$ respectively. The modified CRN $\mathcal{R}_1$ is:
	\begin{equation}\label{eq:ex_extrinsic1_reactions}
		\schemestart
		\subscheme{$\emptyset$}
    		\arrow(z--x1){<=>[$u_1 k_1$][$k_2$]}[0]
		\subscheme{$\mathrm{X}_1$}
		\arrow(@x1--x11){<-[$k_3$]}[0]
		\subscheme{$2\mathrm{X}_1$}
	\schemestop
	\end{equation}
where in this example $\vec{u} = u_1 \in U = \left\{0,1,2\right\}$. Here, $\vec{g}(\vec{u}) = \begin{bmatrix} u_1 & 1 & 1\end{bmatrix}^T$ since the copy number directly scales the rate constant of the production reaction, but does not change the rate constants of the degradation reactions. $\rho(\vec{u})$ takes values of $1/2$, $1/4$, and $1/4$ when $\vec{u}$ is $0$, $1$, and $2$ respectively, which reflects the probabilities of the different copy numbers. The stationary distribution of $(\mathcal{R}_1, \vec{g}(\vec{u}),\rho(\vec{u}),U)$ is then given by the mixture model
\begin{equation}
	f_{X}(\vec{x};\vec{k}) =  \frac{1}{2} v(\vec{x};\mathcal{R}_1((0,k_2,k_3))) + \frac{1}{4} v(\vec{x};\mathcal{R}_1((k_1,k_2,k_3))) + \frac{1}{4}v(\vec{x};\mathcal{R}_1((2k_1,k_2,k_3))).
\end{equation}
\end{example}

We now formally define our notion of identifiability for CRNs with extrinsic noise.
\begin{definition}\label{def:param_id_ext}
	A CRN with extrinsic noise $(\mathcal{R}, \vec{g}(\vec{u}), \rho(\vec{u}),U)$ is \emph{stationary globally identifiable} over $K\subseteq \mathbb{R}^r_{>0}$ if for any $\vec{k}_1,\vec{k}_2\in K$ such that the stationary distribution given by \eqref{eq:extrinsic_noise_stationary_dist} is identical for $\vec{k}=\vec{k}_1$ and $\vec{k}=\vec{k}_2$, there exists $a \in \mathbb{R}$ such that $\vec{k}_2 = a\vec{k}_1$.
\end{definition}
\begin{remark}
	Definition \ref{def:param_id_ext} is the same as Definition \ref{def:param_id} with the exception that Definition \ref{def:param_id_ext} applies to the tuple $(\mathcal{R}, \vec{g}(\vec{u}), \rho(\vec{u}),U)$ that defines a CRN with extrinsic noise. We explicitly give Definition \ref{def:param_id_ext} to emphasize the point that $\vec{g}(\vec{u})$, $\rho(\vec{u})$ and $U$ play a role in determining whether a CRN with extrinsic noise is stationary globally identifiable.
\end{remark}
We now develop a characterization of identifiability in the sense of Definition \ref{def:param_id_ext}. To do this we must deal with the fact that from an observed Gaussian mixture, e.g. of the form \eqref{eq:extrinsic_noise_stationary_dist}, one can only determine the mixture components. This implies that to estimate $\vec{k}$ from the observed distribution we must deal with the problem of not knowing \emph{a priori} which component in the mixture distribution corresponds to each value of $\vec{u}\in U$. Additionally, if $\mathcal{R}\left(\vec{g}(\vec{u})\odot\vec{k}\right)$ is the same for two values of $\vec{u}\in U$, there will be fewer that $|U|$ components identified in the mixture. We begin by formalizing the mapping from a distribution of the form \eqref{eq:extrinsic_noise_stationary_dist} to the set of mixture components. Let $U = \{\vec{u}^1,\vec{u}^2,\dots,\vec{u}^{|U|}\}$. Consider any distribution $f(\vec{x}) = f(\vec{x};\vec{k})$ of the form \eqref{eq:extrinsic_noise_stationary_dist}. Here our notation reinforces the fact that every distribution of this form is generated by some $\vec{k}\in K$, but when solving the identification problem, the value of $\vec{k}\in K$ is initially unknown. We define $C = \mathcal{C}(f(\cdot)) = \left\{\left(w_1,\vec{x}_1,P_1\right),\left(w_2,\vec{x}_2,P_2\right),\dots,\left(w_s,\vec{x}_s,P_s\right)\right\}$ as the smallest set such that
\begin{equation}\label{eq:ext_noise_C_def}
	\forall \vec{x} \in \mathbb{R}^n,\;f(\vec{x}) = \sum_{i=1}^{|U|} \rho(\vec{u}^i)v(\vec{x};(\vec{x}_i,\textstyle\frac{1}{\Omega}P_i)) = \sum_{i=1}^m w_i v(\vec{x}; (\vec{x}_i,\frac{1}{\Omega}P_i)).
\end{equation}
Such a function $\mathcal{C}$ exists by the \emph{uniqueness of representation property} of finite mixtures of Gaussian distributions~\cite{yakowitz1968identifiability}. Conversely, given $C = \mathcal{C}(f(\cdot))$, it is clear that $f(\cdot)$ can be determined uniquely. We note that our use of $f(\cdot)$ as the argument of $\mathcal{C}$ reinforces the fact that $C = \mathcal{C}(f(\cdot))$ is a function of the whole distribution.

\begin{remark}
	Technically, \cite{yakowitz1968identifiability} tells us that $\bar{\mathcal{C}}(f(\cdot))$ defined as the smallest set
	\begin{equation}
		\bar{C} = \bar{\mathcal{C}}(f(\cdot)) = \left\{\left(w_1,\vec{x}_1,\frac{1}{\Omega}P_1\right),\left(w_2,\vec{x}_2,\frac{1}{\Omega}P_2\right),\dots,\left(w_s,\vec{x}_s,\frac{1}{\Omega}P_s\right)\right\}
	\end{equation}		
	 such that
	\begin{equation}
		\forall \vec{x} \in \mathbb{R}^n,\; f(\vec{x};\vec{k}) = \sum_{i=1}^{|U|} \rho(\vec{u}^i)v(\vec{x};(\vec{x}_i,\textstyle\frac{1}{\Omega}P_i))
	\end{equation}
	exists, i.e. from the population distribution we can uniquely identify the mixture components. However, since the mapping between $\bar{C}$ and $C$ is bijective, $\mathcal{C}$ exists and is invertible.
\end{remark}

We now formalize the notion of an assignment of the elements of $C= \mathcal{C}(f(\cdot))$ to the elements of $U$. In general,  for identifiability we need to determine the ``correct'' assignment as well as the true value of $\vec{k}$ from $C= \mathcal{C}(f(\cdot))$. Given a CRN with extrinsic noise $(\mathcal{R},\vec{g}(\vec{u}),\rho(\vec{u}),U)$, for any $f(\cdot)$ of the form \eqref{eq:extrinsic_noise_stationary_dist} with $\vec{k} \in K$ we define $\vec{\sigma}: \{1,2,\dots,|U|\} \rightarrow \mathcal{C}(f(\cdot))$, i.e. a mapping from the indices of the elements of $U$ to the mixture components. We denote $\vec{\sigma}(i) = (\sigma_\rho(i),\sigma_{\vec{x}}(i),\sigma_P(i))$ where for each $i\in\{1,2,\dots,|U|\}$, $(\sigma_\rho(i),\sigma_{\vec{x}}(i),\sigma_P(i)) = (w_j,\vec{x}_j,P_j) \in \mathcal{C}(f(\cdot))$ for some $j$. Given $f(\cdot)$, only some mappings $\vec{\sigma}$ are ``consistent'' with $C$ in the sense that 
\begin{equation}
	\forall \vec{x} \in \mathbb{R}^n,\; \sum_{i=1}^{|U|} \sigma_{\rho}(i) v(\vec{x};(\sigma_{\vec{x}}(i),\sigma_P(i))) = f(\vec{x}).
\end{equation}
The set of consistent $\vec{\sigma}$'s is given by
\begin{equation}
		\Sigma_f  = \{ \vec{\sigma}:\{1,2,\dots,|U|\}\rightarrow \mathcal{C}(f(\cdot)) \text{ surjective} | \sigma_{\rho}(i) = \sum_{j:(\sigma_{\vec{x}}(j),\sigma_{P}(j)) = (\sigma_{\vec{x}}(i),\sigma_{P}(i))}\rho(\vec{u}^j) \}.
\end{equation}

Given a CRN with extrinsic noise $(\mathcal{R},\vec{g}(\vec{u}),\rho(\vec{u}),U)$, for any $f(\vec{x}) = f_X(\vec{x};\vec{k})$ of the form \eqref{eq:extrinsic_noise_stationary_dist} and $\vec{\sigma}\in \Sigma_{f_X(\cdot;\vec{k})}$, we define
\begin{equation}\label{eq:extrinsic_calA_bar}
	\bar{\mathcal{A}}(f(\cdot), \vec{\sigma}) = \begin{bmatrix}
		A(\sigma_{\vec{x}}(1),\sigma_{P}(1))\diag(\vec{g}(\vec{u}^{1})) \\
		A(\sigma_{\vec{x}}(2),\sigma_{P}(2))\diag(\vec{g}(\vec{u}^{2})) \\
		\vdots \\
		A(\sigma_{\vec{x}}(|U|),\sigma_{P}(|U|))\diag(\vec{g}(\vec{u}^{|U|}))
	\end{bmatrix}.
\end{equation}
We then have that $\forall \vec{k} \in K$, $f_X(\cdot;\vec{k})$ satisfies
\begin{equation}
	0 = \bar{\mathcal{A}}(f_{X}(\cdot;\vec{k}), \vec{\sigma}^*)\vec{k}
\end{equation}
where $\vec{\sigma}^*\in \Sigma_{f_{X}(\cdot;\vec{k})}$ satisfies
\begin{equation}	
	\forall i = 1,2,\dots,|U|,\; \left(\sigma_{\vec{x}}(i),\sigma_{P}(i)\right) = \mathcal{R}(\vec{g}(\vec{u}^i)\odot \vec{k})
\end{equation}

\begin{lemma}\label{lem:ext_basic}
	A CRN with extrinsic noise $(\mathcal{R},\vec{g}(\vec{u}),\rho(\vec{u}),U)$, is stationary globally identifiable if for all $f(\vec{x}) = f(\vec{x};\vec{k})$ of the form \eqref{eq:extrinsic_noise_stationary_dist}, there exists $\vec{\xi}\in\mathbb{R}^r$ such that for all $(\vec{\sigma},\vec{k})\in (\Sigma_f,K)$ satisfying $0 = \bar{\mathcal{A}}(f(\cdot),\vec{\sigma})\vec{k}$, $\vec{k} = a \vec{\xi}$ for some $a \in \mathbb{R}$.
\end{lemma}

\begin{proof}
	We prove the contrapositive. To begin, suppose that $(\mathcal{R},\vec{g}(\vec{u}),\rho(\vec{u}),U)$ is not stationary globally identifiable. Then, there exists $f(\cdot)$ and $\vec{k}',\vec{k}''\in K$ with $\vec{k}' \neq \alpha \vec{k}''$ for any $\alpha$ such that
	\begin{equation}
		f(\cdot) = \sum_{i=1}^{|U|} \rho(\vec{u}^i)v(\cdot;\mathcal{R}(\vec{g}(\vec{u}^i)\odot \vec{k}'))
	\end{equation}
	and
	\begin{equation}
		f(\cdot) = \sum_{i=1}^{|U|} \rho(\vec{u}^i)v(\cdot;\mathcal{R}(\vec{g}(\vec{u}^i)\odot \vec{k}'')).
	\end{equation}
	Let us define $\vec{\sigma}'$ by $\vec{\sigma}'(i) = (\sigma'_{\rho}(i), \sigma_{\vec{x}}'(i), \sigma_{P}'(i))$ where $(\sigma_{\vec{x}}'(i), \sigma_{P}'(i)) = \mathcal{R}(\vec{g}(\vec{u}^i)\odot \vec{k}')$ and 
	\begin{equation}
		\sigma_{\rho}'(i) = \sum_{j: \mathcal{R}(\vec{g}(\vec{u}^j)\odot \vec{k}) =  \mathcal{R}(\vec{g}(\vec{u}^i)\odot \vec{k})}\rho(\vec{u}^j).
	\end{equation}
	Similarly, we define  $\vec{\sigma}''$ by $\vec{\sigma}''(i) = (\sigma''_{\rho}(i), \sigma_{\vec{x}}''(i), \sigma_{P}''(i))$ where $(\sigma_{\vec{x}}''(i), \sigma_{P}''(i)) = \mathcal{R}(\vec{g}(\vec{u}^i)\odot \vec{k}'')$ and 
	\begin{equation}
		\sigma_{\rho}''(i) = \sum_{j: \mathcal{R}(\vec{g}(\vec{u}^j)\odot \vec{k}'') =  \mathcal{R}(\vec{g}(\vec{u}^i)\odot \vec{k}'')}\rho(\vec{u}^j).
	\end{equation}
	Observe that $\vec{\sigma}',\vec{\sigma}'' \in \Sigma_f$. We have
	\begin{equation}
		\bar{\mathcal{A}}(f(\cdot), \vec{\sigma}') = 
		\begin{bmatrix}
			A(\sigma_{\vec{x}}'(1),\sigma_{P}'(1))\diag(\vec{g}(\vec{u}^{1})) \\
			A(\sigma_{\vec{x}}'(2),\sigma_{P}'(2))\diag(\vec{g}(\vec{u}^{2})) \\
			\vdots \\
			A(\sigma_{\vec{x}}'(|U|),\sigma_{P}'(|U|))\diag(\vec{g}(\vec{u}^{|U|}))
		\end{bmatrix},
	\end{equation}
	and furthermore, for all $i \in \{1,2,\dots,|U|\}$, since 
	\begin{equation}
		(\sigma_{\vec{x}}'(i),\sigma_{P}'(i)) = \mathcal{R}(\vec{g}(\vec{u}^i\odot\vec{k}'),
	\end{equation}
	we have that $0 = A(\sigma_{\vec{x}}'(i),\sigma_{P}'(i))\diag(\vec{g}(\vec{u}^{i}))\vec{k}'$. Therefore, $0 = \bar{\mathcal{A}}(f(\cdot), \vec{\sigma}')\vec{k}'$. Similarly, $0 = \bar{\mathcal{A}}(f(\cdot), \vec{\sigma}'')\vec{k}''$. Therefore, it is not the case that for all $(\vec{\sigma},\vec{k})\in (\Sigma_f,K)$ satisfying $0 = \bar{\mathcal{A}}(f(\cdot),\vec{\sigma})\vec{k}$, $\vec{k} = a \vec{\xi}$ for some $a \in \mathbb{R}$, which completes our proof.
\end{proof}

\begin{condition}\label{cond:sigma_unique_cond}
	The CRN with extrinsic noise $(\mathcal{R},\vec{g}(\vec{u}),\rho(\vec{u}),U)$ is such that for all $f(\vec{x}) = f(\vec{x};\vec{k})$ of the form \eqref{eq:extrinsic_noise_stationary_dist}, there exists a unique $\vec{\sigma}^f \in \Sigma_f$ such that $0 = \bar{\mathcal{A}}(f(\cdot),\vec{\sigma}^f)\vec{k}$ for some $\vec{k}\in K$.
\end{condition}

\begin{lemma}\label{lem:uniqe_sigma}
	A CRN with extrinsic noise $(\mathcal{R},\vec{g}(\vec{u}),\rho(\vec{u}),U)$, is identifiable if it satisfies Condition \ref{cond:sigma_unique_cond}, and furthermore, for all $f(\cdot)$ of the form \eqref{eq:extrinsic_noise_stationary_dist}, 
	\begin{equation}\label{eq:ext_noise_lem2_rank}
		\rank \bar{\mathcal{A}}(f(\cdot),\vec{\sigma}^f) = r-1.
	\end{equation}		
	Here $\vec{\sigma}^f$ is the unique $\vec{\sigma} \in \Sigma_f$ such that $0 = \bar{\mathcal{A}}(f(\cdot),\vec{\sigma})\vec{k}$ for some $\vec{k}\in K$.
\end{lemma}
\begin{proof}
	The result follows from Lemma \ref{lem:ext_basic}. For any $f(\cdot)$ of the form \eqref{eq:extrinsic_noise_stationary_dist}, assumption 1) ensures that all solutions $(\vec{\sigma},\vec{k})$ to $0 = \bar{\mathcal{A}}(f(\cdot),\vec{\sigma})\vec{k}$ are of the form $(\vec{\sigma}^f, \vec{k})$ for some $\vec{k}$. Assumption 2) then ensures that the dimension of the nullspace of $\bar{\mathcal{A}}(f(\cdot),\vec{\sigma}^f)$ is one, and hence $\exists \vec{v}\in K$ such that $0 = \bar{\mathcal{A}}(f(\cdot),\vec{\sigma}^f)\vec{k}$ if and only if $\vec{k} = \alpha \vec{v}$ for some $\alpha$.
\end{proof}
We now develop a criteria for identifiability that is amenable to analysis using algebraic tools of Section \ref{sec:nullstellensatz}. Given a CRN with extrinsic noise $(\mathcal{R}, \vec{g}(\vec{u},\rho(\vec{u}),U)$, we define $\bar{A}: \left(\mathbb{R}^n \times \mathbb{S}^{n\times n}\right)^{|U|} \rightarrow \mathbb{R}^{|U|(\frac{n^2+3n)}{2} \times r}$ by
\begin{equation}\label{eq:extrinsic_noise_A_bar}
		\bar{A}((\vec{x}_1,P_1),(\vec{x}_2,P_2),\dots,(\vec{x}_{|U|},P_{|U|})) = \begin{bmatrix}
		A(\vec{x}_{1},P_{1})\diag(\vec{g}(\vec{u}^{1})) \\
		A(\vec{x}_{2},P_{2})\diag(\vec{g}(\vec{u}^{2})) \\
		\vdots \\
		A(\vec{x}_{|U|},P_{|U|})\diag(\vec{g}(\vec{u}^{|U|}))
	\end{bmatrix}.
\end{equation}

\begin{theorem}\label{thm:ext_basic_check}
	Consider a CRN with extrinsic noise $(\mathcal{R},\vec{g}(\vec{u}),\rho(\vec{u}),U)$. If Condition \ref{cond:sigma_unique_cond} holds and for all
	\begin{equation}
		\left((\vec{x}_1,P_1),(\vec{x}_2,P_2),\dots,(\vec{x}_{|U|},P_{|U|})\right) \in \left(\mathbb{R}^n_{\geq 0} \times \mathbb{S}^{n\times n}\right)^{|U|}
	\end{equation}
	such that there exists $\vec{k} \in K$ satisfying $0 = \bar{A}((\vec{x}_1,P_1),(\vec{x}_2,P_2),\dots,(\vec{x}_{|U|},P_{|U|}))\vec{k}$, we have
	\begin{equation}\label{eq:ext_noise_thm1_rank}
		\rank \bar{A}((\vec{x}_1,P_1),(\vec{x}_2,P_2),\dots,(\vec{x}_{|U|},P_{|U|})) \geq r-1,
	\end{equation}
	then $(\mathcal{R},\vec{g}(\vec{u}),\rho(\vec{u}),U)$ is stationary globally identifiable over $K$.
\end{theorem}
\begin{proof}
	To apply Lemma \ref{lem:uniqe_sigma} we must show that the rank condition \eqref{eq:ext_noise_thm1_rank} implies assumption \eqref{eq:ext_noise_thm1_rank} of Lemma \ref{lem:uniqe_sigma}. Let $f(\cdot)$ be of the form \eqref{eq:extrinsic_noise_stationary_dist}. We have that
	\begin{equation}
		\bar{\mathcal{A}}(f(\cdot),\vec{\sigma}^f) =
		\begin{bmatrix}
		A(\sigma^f_{\vec{x}}(1),\sigma^f_{P}(1))\diag(\vec{g}(\vec{u}^{1})) \\
		A(\sigma^f_{\vec{x}}(2),\sigma^f_{P}(2))\diag(\vec{g}(\vec{u}^{2})) \\
		\vdots \\
		A(\sigma^f_{\vec{x}}(|U|),\sigma^f_{P}(|U|))\diag(\vec{g}(\vec{u}^{|U|}))
	\end{bmatrix}.
	\end{equation}
	Observe that for all $i \in \{1,2,\dots,|U|\}$, $(\sigma^f_{\vec{x}}(i),\sigma^f_{P}(i)) \in \left(\mathbb{R}^n_{\geq 0} \times \mathbb{S}^{n\times n}\right)$. Therefore, 
	\begin{equation}
		\bar{\mathcal{A}}(f(\cdot),\vec{\sigma}^f) = \bar{A}((\sigma^f_{\vec{x}}(1),\sigma^f_{P}(1)),(\sigma^f_{\vec{x}}(2),\sigma^f_{P}(2)),\dots,(\sigma^f_{\vec{x}}(|U|),\sigma^f_{P}(|U|))).
	\end{equation}
	Hence, by \eqref{eq:ext_noise_thm1_rank}, $\rank \bar{\mathcal{A}}(f(\cdot),\vec{\sigma}^f) \geq r-1$. Furthermore, the fact that Condition \ref{cond:sigma_unique_cond} holds ensures that  $\rank \bar{\mathcal{A}}(f(\cdot),\vec{\sigma}^f) \leq r-1$, and so $\rank \bar{\mathcal{A}}(f(\cdot),\vec{\sigma}^f) = r-1$. By applying Lemma \ref{lem:uniqe_sigma} we then obtain the desired result.
\end{proof}

Theorem \ref{thm:ext_basic_check} can be turned into an algebraic condition for identifiability that can be checked computationally. However, in general, it is hard to check that Condition \ref{cond:sigma_unique_cond} holds. Therefore, we now focus on a special case which occurs frequently in synthetic biology where Condition \ref{cond:sigma_unique_cond} is guaranteed to hold. To begin this investigation we define the \emph{augmented CRN} of a CRN with extrinsic noise as follows.

\begin{definition}\label{def:aug_CRN}
	Given a CRN with extrinsic noise, $(\mathcal{R}, \vec{g}(\vec{u}), \rho(\vec{u}),U)$ and $\vec{\alpha}\in \mathbb{R}_{>0}^s, \gamma >0$, we define the augmented version of the CRN $\mathcal{R}_{aug}$, as the CRN with species $\mathrm{X_1}, \dots, \mathrm{X}_n$ from $\mathcal{R}$ along with species $\mathrm{Z}_1,\dots,\mathrm{Z}_s$, and all reactions from $\mathcal{R}$ along with
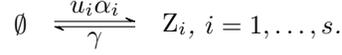
\begin{center}
\schemestart
	\subscheme{$\emptyset$}
    	\arrow(O--Zi){<=>[$u_i\alpha_i$][$\gamma$]}[0]
	\subscheme{$\mathrm{Z}_i$}
\schemestop, $i = 1,\dots,s$.
\end{center}
Here we recall that $s$ is the dimension of $\vec{u}$. We denote the augmented version of a CRN $(\mathcal{R},\vec{g}(\vec{u}),\rho(\vec{u}),U)$ with parameters $\vec{\alpha}$ and $\gamma$ by $(\mathcal{R}_{aug}, \vec{g}_{aug}(\vec{u}),\rho(\vec{u}),U, \vec{\alpha}, \gamma)$.
\end{definition}

\begin{remark}
	The ideas we have developed for CRNs with extrinsic noise apply to augmented CRNs as well. In fact, for a fixed value of $\vec{\alpha}$ and $\gamma$, Definition \ref{def:param_id_ext} can be applied to an augmented CRN with extrinsic noise, since $(\mathcal{R}_{aug}, \vec{g}_{aug}(\vec{u}),\rho(\vec{u}),U, \vec{\alpha}, \gamma)$ defines a map from $\vec{k}$ to a Gaussian mixture model. Theorem \ref{thm:ext_basic_check} can be used for an augmented CRN $(\mathcal{R}, \vec{g}(\vec{u}), \rho(\vec{u}),U)$. In this case the $\bar{A}$ used in Theorem \ref{thm:ext_basic_check}, and the $\bar{\mathcal{A}}(f(\cdot),\vec{\sigma})$ used in Lemma \ref{lem:uniqe_sigma} are the same as $\bar{A}$ and $\bar{\mathcal{A}}$ defined for the non-augmented CRN $(\mathcal{R}, \vec{g}(\vec{u}), \rho(\vec{u}),U)$. This is due to the fact that the only reactions involving the $Z$ species have rate constants $\vec{\alpha}$ or $\gamma$, which are known constants, and thus do not need to be inferred from the stationary distribution.
\end{remark}
\begin{remark}
	In applications in synthetic biology it is often the case that one has an augmented CRN in the sense of Definition \ref{def:aug_CRN}. One example is when a biomolecular circuit is constructed on one or more plasmids which are transformed in the cells and each plasmid has a constitutive reporter. Each constitutive reporter is a fluorescent protein whose amount is proportional to the copy number of the plasmid. Additionally, it is possible to estimate $\vec{\alpha}$ and $\gamma$ in a separate experiment where the copy number is well controlled  \cite{delahoz2000plasmid}. Note that the reaction rate constant vector of $(\mathcal{R}_{aug}, \vec{g}_{aug}(\vec{u}),\rho(\vec{u}),U)$ is the same as that of $(\mathcal{R},\vec{g}(\vec{u}),\rho(\vec{u}),U)$, and we treat $\vec{\alpha}$ and $\gamma$ as known constants.
\end{remark}

The following continuation of Example \ref{ex:1d_extrinsic} illustrates Definition \ref{def:aug_CRN}.
\setcounter{example}{9}
\begin{example}[1-dimensional extrinsic noise]
	Continuing with Example \ref{ex:1d_extrinsic}, we now consider the case where there is a constitutive reporter in the circuit. The augmented CRN $(\mathcal{R}_{1aug},\vec{g}_{aug}(\vec{u}),\rho(\vec{u}),U,\vec{\alpha},\gamma)$ is given by
	\begin{equation}
	\schemestart
		\subscheme{$\mathrm{Z}_1$}
		\arrow(Z--z){<=>[$\gamma$][$u_{1}\alpha_1$]}[0]
		\subscheme{$\emptyset$}
    		\arrow(@z--x1){<=>[$u_{1}k_1$][$k_2$]}[0]
		\subscheme{$\mathrm{X}_1$}
		\arrow(@x1--x11){<-[$k_3$]}[0]
		\subscheme{$2\mathrm{X}_1$}
        \schemestop.
\end{equation}
Here $Z_1$ is the constitutive reporter. Its production rate is proportional to the copy number, $\vec{u} = u_1$, which takes a, constant, value drawn from $\rho(\vec{u})$ in each cell.
\end{example}

The augmented version of any CRN will satisfy Condition \ref{cond:sigma_unique_cond}, and thus we can readily construct an algebraic condition that is sufficient for identifiability of augmented CRNs. We formalize this fact in the following theorem.

\begin{theorem}\label{thm:grob_check_constit}%
	Consider a CRN with extrinsic noise $(\mathcal{R},g(\vec{u}),\rho(\vec{u}),U)$. Let $\vec{\alpha}^0 \in \mathbb{R}^s_{>0}$, and let
	\begin{equation}
		\bar{K} = \left\{ (\vec{k},\vec{y}) \in \mathbb{R}^{r+m} \middle| h_i(\vec{k},\vec{y}) = 0, \;i = 1,2,\dots,p \right\}
	\end{equation}
	be a lifted representation of $K$. Let $\left\{u_1,u_2,\dots,u_l\right\} \subseteq U$ and denote row $q$ of $\bar{A}$ by \\$\bar{A}_q(\vec{x}_1,\dots,\vec{x}_{l},P_1,\dots,P_{l}, \vec{u}^1, \dots, \vec{u}^l)$. If the reduced Gr\"{o}bner basis of
	\begin{equation}\label{eq:ideal_ext_constit}
	\begin{multlined}
		\left\langle \vphantom{M_{|U|}^1} h_i(\vec{k},\vec{y}), \;\forall i\in \{1,\dots,p\},\right.
		\bar{A}_q(\vec{x}_1,\dots,\vec{x}_{l},P_1,\dots,P_{l}, \vec{u}^1, \dots, \vec{u}^l)\vec{k},\;\forall q\in \{1,\dots, u\frac{n^2 + 3n}{2}\},\\
		\left. \bar{M}_i^{(r-1)\times (r-1)}(\vec{x}_1,\dots,\vec{x}_{l},P_1,\dots,P_{l}, \vec{u}^1, \dots, \vec{u}^l)\vec{k},\;\forall i\in \{1,\dots,m\}\right\rangle
	\end{multlined}
	\end{equation}
	is $\{1\}$, then the augmented CRN $(\mathcal{R}_{aug},\vec{g}_{aug}(\vec{u}),\rho(\vec{u}),U, \vec{\alpha}^0, 1)$, given in Definition \ref{def:aug_CRN}, is stationary globally identifiable over $K$.
\end{theorem}

\begin{proof}
	For notational clarity we use $\bar{\mathcal{A}}(f(\cdot),\vec{\sigma})$ refer to the matrix defined by \eqref{eq:extrinsic_calA_bar} for the CRN $(\mathcal{R},g(\vec{u}),\rho(\vec{u}),U)$, and $\bar{\mathcal{A}}_{aug}(f(\cdot),\vec{\sigma})$ refer to the matrix defined by \eqref{eq:extrinsic_calA_bar} for the augmented CRN $(\mathcal{R}_{aug},\vec{g}_{aug}(\vec{u}),\rho(\vec{u}),U,\vec{\alpha},\gamma)$. Observe that $\bar{\mathcal{A}}_{aug}(f(\cdot),\vec{\sigma})$ is used to determine if the augmented CRN satisfies Condition \ref{cond:sigma_unique_cond}, whereas $\bar{\mathcal{A}}(f(\cdot),\vec{\sigma})$ determines identifiability of the augmented CRN. This is due to $\vec{\alpha}$ and $\gamma$ being known constants instead of parameters that must be estimated. We partition $P$ as
	\begin{equation}
		P = \begin{bmatrix}
			P_{\vec{x}} & P_{\vec{x},\vec{z}}\\
			P_{\vec{x},\vec{z}}^T & P_{\vec{z}}
		\end{bmatrix}.
	\end{equation}
	Observe that $\bar{\mathcal{A}}_{aug}(f(\cdot),\vec{\sigma})$ takes the form
\begin{equation}\label{eq:ext_A_bar_aug_form}
	 \bar{\mathcal{A}}_{aug}(f(\cdot),\vec{\sigma}) = 
	 \mbox{\small$
	 \begin{bmatrix}
		A(\sigma_{\vec{x}}'(1),\sigma_{P_{\vec{x}}}'(1))\diag(\vec{g}(\vec{u}^{1})) & 0 & 0\\
		0 & \diag(\vec{u}^1) & -\sigma_{\vec{z}}(\vec{u}^1)\\
		0 & 2\diag(\sigma_{P_{\vec{z}}}(\vec{u}^1))-\diag(\vec{u}^1) & -\sigma_{\vec{z}}(\vec{u}^1)\\
		A(\sigma_{\vec{x}}'(2),\sigma_{P_{\vec{x}}}'(2))\diag(\vec{g}(\vec{u}^{2})) & 0 & 0\\
		0 & \diag(\vec{u}^2) & -\sigma_{\vec{z}}(\vec{u}^2)\\
		0 & 2\diag(\sigma_{P_{\vec{z}}}(\vec{u}^2))-\diag(\vec{u}^2) & -\sigma_{\vec{z}}(\vec{u}^2)\\
		\vdots & \vdots & \vdots\\
		A(\sigma_{\vec{x}}'(|U|),\sigma_{P_{\vec{x}}}'(|U|))\diag(\vec{g}(\vec{u}^{|U|})) & 0 & 0\\
		0 & \diag(\vec{u}^{|U|}) & -\sigma_{\vec{z}}(\vec{u}^{|U|})\\
		0 & 2\diag(\sigma_{P_{\vec{z}}}(\vec{u}^{|U|}))-\diag(\vec{u}^{|U|}) & -\sigma_{\vec{z}}(\vec{u}^{|U|})\\
	 \end{bmatrix}$}.
\end{equation}
	We use Theorem \ref{thm:ext_basic_check} to prove the desired result. To do so we must show that Condition \ref{cond:sigma_unique_cond} holds for $(\mathcal{R}_{aug},\vec{g}_{aug}(\vec{u}),\rho(\vec{u}),U)$. Suppose that there exists $\vec{\sigma}^1, \vec{\sigma}^2 \in \Sigma_f$ such that $\vec{\sigma}^1 \neq \vec{\sigma}^2$ and
	\begin{subequations}
	\begin{align}
		0 &= \bar{\mathcal{A}}_{aug}(f(\cdot),\vec{\sigma}^1)
		\begin{bmatrix}
			\vec{k}^1 \\
			\vec{\alpha}^0\\
			\gamma
		\end{bmatrix}\\
		0 &= \bar{\mathcal{A}}_{aug}(f(\cdot),\vec{\sigma}^2)
		\begin{bmatrix}
			\vec{k}^2 \\
			\vec{\alpha}^0\\
			\gamma
		\end{bmatrix}
	\end{align}
	\end{subequations}
	with $\vec{k}^1,\vec{k}^2 \in K$, $\vec{\alpha}^0 \in \mathbb{R}_{> 0}^s$, and $\gamma = 1$. Then, from \eqref{eq:ext_A_bar_aug_form} we have that for all $i = 1,2,\dots,|U|$,
	\begin{subequations}
	\begin{align}
		0 &= \vec{\alpha}^0 \odot \vec{u}^i - \sigma_{\vec{z}}^1(i),\\
		0 &= \vec{\alpha}^0 \odot \vec{u}^i - \sigma_{\vec{z}}^2(i).
	\end{align}
	\end{subequations}
	This implies that for all $i = 1,2,\dots,|U|$, we have that $\sigma_{\vec{z}}^1(i) = \sigma_{\vec{z}}^2(i)$. Therefore, $|\mathcal{C}(f(\cdot))| \geq |U|$. Additionally, we know that it always holds that $|\mathcal{C}(f(\cdot))| \leq |U|$. Therefore, we can then infer that $|\mathcal{C}(f(\cdot))| = |U|$. Thus, $\sigma_{\vec{z}}^1(i) = \sigma_{\vec{z}}^2(i)$ for $i = 1,2,\dots,|U|$ implies that $\vec{\sigma}^1(i) = \vec{\sigma}^2(i)$ for $i=1,2,\dots,|U|$. This shows that only one $\vec{\sigma} \in \Sigma^f$ has a $\vec{k} \in K$ such that $0 = \bar{\mathcal{A}}(f(\cdot),\vec{\sigma})\vec{k}$ for some $\vec{k}\in K$, and therefore Condition \ref{cond:sigma_unique_cond} is satisfied by $(\mathcal{R}_{aug},\vec{g}_{aug}(\vec{u}),\rho(\vec{u}),U, \vec{\alpha},\gamma)$. To complete the proof, observe that \eqref{eq:ideal_ext_constit} being equal to $\{1\}$ ensures that Theorem \ref{thm:ext_basic_check} can be applied, and so $(\mathcal{R}_{aug},\vec{g}_{aug}(\vec{u}),\rho(\vec{u}),U, \vec{\alpha}^0, 1)$, is stationary globally identifiable over $K$.
\end{proof}

\begin{remark}
	We note that Condition \ref{cond:sigma_unique_cond} is needed for the emptiness of the ideal defined by \eqref{eq:ideal_ext_constit} to be a sufficient condition for stationary global identifiability of $(\mathcal{R}, \vec{g}(\vec{u}), \rho(\vec{u}), U)$. This is because without Condition \ref{cond:sigma_unique_cond} there are two ways for a CRN with extrinsic noise to lose identifiability: a) There is exactly one $\vec{\sigma}$ consistent with $f(\cdot)$ and $(\mathcal{R},g(\vec{u}),\rho(\vec{u}),U)$, but $\rank \bar{\mathcal{A}}(f(\cdot),\vec{\sigma}) < r-1$, which is analogous to the loss of identifiability for CRNs without extrinsic noise, or b) There are multiple $\vec{\sigma}$'s consistent with $f(\cdot)$ and $(\mathcal{R},g(\vec{u}),\rho(\vec{u}),U)$, and each corresponds to a different 1-dimensional subspace of for $\vec{k}$. In Theorem \ref{thm:grob_check_constit} we use the fact that the augmented CRN is considered to ensure that Condition \ref{cond:sigma_unique_cond} holds.
\end{remark}

\begin{remark}
	We note that identifiability in sense that Theorem \ref{thm:grob_check_constit} certifies assumes that both $\vec{\alpha}$ and $\gamma$ are known, with $\gamma = 1$. However, since this work studies only stationary distributions, as long as $\vec{\alpha}/\gamma$ is known we can always take $\gamma = 1$ and use the value of $\vec{\alpha}/\gamma$ in place of $\vec{\alpha}$.
\end{remark}

\setcounter{example}{9}
\begin{example}[1-dimensional extrinsic noise]
Here we continue Example \ref{ex:1d_extrinsic}. Let $\vec{\alpha} > 0$. We wish to certify identifiability of $(\mathcal{R}_{aug}, \vec{g}_{aug}(\vec{u}),U, \vec{\alpha}, 1)$ over $\mathbb{R}^2_{>0}$. Theorem \ref{thm:grob_check_constit} states that we can consider the ideal \eqref{eq:ideal_ext_constit}, and if the reduced Gr\"{o}bner basis is $\{1\}$, we can conclude that stationary global identifiability holds. For this example, \eqref{eq:ideal_ext_constit} is defined by 54 polynomials.%
\end{example}

We observe that if we want to use Theorem \ref{thm:grob_check_constit} to certify stationary global identifiability we must compute the reduced Gr\"{o}bner basis of an ideal over $\mathbb{Q}[[\vec{x}^T,\vec{y}^T,\vec{k}^T]^T]$. If for example $K = \mathbb{R}_{>0}^r$, then $[\vec{x}^T,\vec{y}^T,\vec{k}^T]^T \in \mathbb{R}^{l\frac{n^2 + 3n}{2} + r}$, and hence as $|U|$ grows our computational problem becomes harder very quickly, since we may need to use $l=|U|$ in the worst case. An alternative is to use only the reaction rate equations \eqref{eq:LNA_rre}, which conceptually equates to using only the means of each mixture component in the estimation of the parameters. Let $A^{rre}(\vec{x})$ be the first $n$ rows of $A(\vec{x},P)$, and for any $l \leq |U|$, define
\begin{equation}
	\bar{A}^{rre}(\vec{x}_1,\dots,\vec{x}_{l},\vec{u}^1, \dots, \vec{u}^l) =
	\begin{bmatrix}
		A^{rre}(\vec{x}_1)\diag(\vec{g}(\vec{u}^1)) \\
		A^{rre}(\vec{x}_2)\diag(\vec{g}(\vec{u}^2)) \\
		\vdots \\
		A^{rre}(\vec{x}_{l})\diag(\vec{g}(\vec{u}^{l}))
	\end{bmatrix}.
\end{equation}
Since the first $n$ rows of $A(\vec{x},P)$ correspond to the reaction rate equations \eqref{eq:stationary_dist_rre} they are not a function of $P$, and therefore neither is $\bar{A}^{rre}$. Therefore, we can eliminate all the covariance variables from \eqref{eq:ideal_ext_constit} which results in a check for stationary parametric identifiability involving an ideal over a lower dimensional ring.
\begin{theorem}\label{thm:ext_noise_cert_rre_only}
	Consider a CRN with extrinsic noise $(\mathcal{R},g(\vec{u}),\rho(\vec{u}),U)$. Let $\vec{\alpha}^0 \in \mathbb{R}^s_{>0}$, and let
	\begin{equation}
		\bar{K} = \left\{ (\vec{k},\vec{y}) \in \mathbb{R}^{r+m} \middle| h_i(\vec{k},\vec{y}) = 0, \;i = 1,2,\dots,p \right\}
	\end{equation}
	be a lifted representation of $K$. Let $\left\{u_1,u_2,\dots,u_l\right\} \subseteq U$. Denote by $\bar{A}^{rre}_q$ row $q$ of $\bar{A}^{rre}$, and denote by 
	\begin{equation}
		\bar{M}_i^{rre,(r-1)\times (r-1)}(\vec{x}_1,\dots,\vec{x}_{l}, \vec{u}^1, \dots, \vec{u}^l)
	\end{equation}
	the $(r-1)\times (r-1)$ minors of $\bar{A}^{rre}$, indexed by $i$. If the reduced Gr\"{o}bner basis of 
		\begin{equation}\label{eq:rre_only_ideal}
		\begin{multlined}
		\left\langle \vphantom{M_{|U|}^1} h_i(\vec{k},\vec{y}), \;\forall i \in \{1,\dots,p\},\right.
		\bar{A}^{rre}_q(\vec{x}_1,\dots,\vec{x}_{l},\vec{u}^1, \dots, \vec{u}^l)\vec{k},\;\forall q \in \{1,\dots, un\}\\
		\left. \bar{M}_i^{rre,(r-1)\times (r-1)}(\vec{x}_1,\dots,\vec{x}_{l}, \vec{u}^1, \dots, \vec{u}^l)\vec{k},\;\forall i\in \{1,\dots,m\} \right\rangle
	\end{multlined}
	\end{equation}
	is $\{1\}$, then the augmented CRN defined in Definition \ref{def:aug_CRN} $(\mathcal{R}_{aug},\vec{g}_{aug}(\vec{u}),\rho(\vec{u}),U, \vec{\alpha}^0, 1)$, is stationary globally identifiable over $K$.
\end{theorem}
\begin{proof}
	We observe that $\rank \bar{A}^{rre} \leq \rank \bar{A}$, and therefore if the reduced Gr\"{o}bner basis of the ideal \eqref{eq:rre_only_ideal} is $\{1\}$, the rank of $\bar{A}$ cannot drop below $r-1$ for any admissible $\vec{x}_1,\dots,\vec{x}_{l},P_1,\dots,P_l$ and hence the ideal \eqref{eq:ideal_ext_constit} has reduced Gr\"{o}bner basis $\{1\}$. Therefore, $(\mathcal{R}_{aug},\vec{g}_{aug}(\vec{u}),\rho(\vec{u}),U, \vec{\alpha}^0, 1)$ is stationary globally identifiable over $K$ by Theorem \ref{thm:grob_check_constit}.
\end{proof}
\setcounter{example}{9}
\begin{example}[1-dimensional extrinsic noise]
We now return to Example \ref{ex:1d_extrinsic}. Suppose we want to certify that $(\mathcal{R}_1,\vec{g}(\vec{u}),\rho(\vec{u}),U,\alpha,1)$ is stationary globally identifiable over $\mathbb{R}_{>0}^3$, while using fewer variables. For this example, $\bar{A}^{rre}$ is given by
\begin{equation}
	\bar{A}^{rre}(\vec{x}_1,\vec{x}_2,\vec{x}_3) = 
		\begin{bmatrix}
		0 & -x_{11} & - x_{11}^2 \\
		1 & -x_{12} & - x_{12}^2 \\
		2 & -x_{13} & - x_{13}^2
	\end{bmatrix}.
\end{equation}
Theorem \ref{thm:ext_noise_cert_rre_only} states that we can consider the ideal \eqref{eq:rre_only_ideal}, and if the reduced Gr\"{o}bner basis is $\{1\}$, we can conclude that stationary global identifiability holds.%
\end{example}

We now present an important example where Theorem \ref{thm:ext_noise_cert_rre_only} can be used to certify stationary global identifiability.
\begin{example}[gaining identifiability by adding extrinsic noise]\label{ex:feedback}
We consider a feedback loop consisting of two species, $\mathrm{X}_1$ and $\mathrm{X}_2$ where as shown in Figure \ref{fig:ext_noise_antithetic} $\mathrm{X}_1$ and $\mathrm{X}_2$ mutually degrade, and $\mathrm{X}_2$ activates the production of $\mathrm{X}_1$. As in Example \ref{ex:which_act} we model the activation of $\mathrm{X}_1$ by $\mathrm{X}_2$ as the production rate of $\mathrm{X}_1$being an affine function, $k_1 + k_6 x_2$. This system forms a conceptual model of a feedback loop with only two species, where as we will see the system is not stationary globally identifiable without extrinsic noise, but is stationary globally identifiable with extrinsic noise. To start, we note that without the extrinsic noise the CRN is not stationary globally identifiable since for the CRN
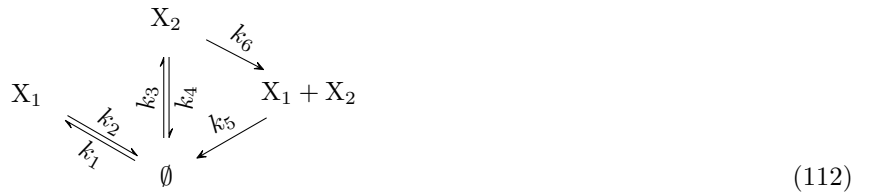
\begin{equation}
	\schemestart
		\subscheme{$\emptyset$}
    		\arrow(z--x1){<=>[$k_2$][$k_1$]}[150]
		\subscheme{$\mathrm{X}_1$}
		\arrow(@z--x2){<=>[$k_3$][$k_4$]}[90]
		\subscheme{$\mathrm{X}_2$}
		\arrow(@z--x12){<-[$k_5$]}[30]
		\subscheme{$\mathrm{X}_1 + \mathrm{X}_2$}
        \arrow(@x2--@x12){->[$k_6$]}
	\schemestop
\end{equation}
we have from the definition of $\vec{f}(\vec{x};\vec{k})$ in \eqref{eq:RRE_def} that
\begin{equation}
	\vec{f}(\vec{x};\vec{k}) =
	\begin{bmatrix}
		k_1 - k_2x_1 - k_5x_1x_2 + k_6x_2\\
		k_3 - k_4x_2 - k_5x_1x_2
	\end{bmatrix}
\end{equation}
and from \eqref{eq:Gamma_def} that
\begin{equation}
	\Gamma(\vec{x};\vec{k})\Gamma(\vec{x};\vec{k})^T =
	\begin{bmatrix}
		k_1 + k_2x_1 + k_6x_2 + k_5x_1x_2 & k_5x_1x_2\\
		k_5x_1x_2 & k_3 + k_4x_2 + k_5x_1x_2
	\end{bmatrix}.
\end{equation}
Therefore we have that \eqref{eq:stationary_A} is given by $0 = A(\vec{x},P)\vec{k}$ where
\begin{equation}\label{eq:fb_example_A}
	A(\vec{x},P) = 
	\begin{bmatrix}
		1 &         -x_1 & 0 &           0 &                                        -x_1x_2 &          x_2\\
		0 &           0 & 1 &         -x_2 &                                        -x_1x_2 &           0\\
		1 & x_1 - 2p_{11} & 0 &           0 &                 x_1x_2 - 2p_{12}x_1 - 2p_{11}x_2 & 2p_{12} + x_2\\
		0 &       -p_{12} & 0 &       -p_{12} & x_1x_2 - p_{12}x_1 - p_{12}x_2 - p_{22}x_1 - p_{11}x_2 &        p_{22}\\
		0 &           0 & 1 & x_2 - 2p_{22} &                 x_1x_2 - 2p_{22}x_1 - 2p_{12}x_2 &           0
	\end{bmatrix}.
\end{equation}
One can verify that when $\vec{k} = \begin{bmatrix} 10 & 1 & 10 & 1 & 1 & 10\end{bmatrix}^T$ the solution to \eqref{eq:fb_example_A} is $\vec{x} = \begin{bmatrix} 10 &  \frac{10}{11} \end{bmatrix}^T$ and
\begin{equation}
	P = \begin{bmatrix}
		10 & 0 \\
		0 & 10/11
	\end{bmatrix}.
\end{equation}
Evaluating the rank of $A$ in \eqref{eq:fb_example_A} with these values of $\vec{x}$ and $P$ gives $\rank A = 4 < r- 1$ and so the CRN without extrinsic noise is not stationary globally identifiable.

We now consider extrinsic noise, where the genes for $\mathrm{X}_1$ and $\mathrm{X}_2$ are on separate plasmids, each with its own constitutive reporter, $\mathrm{X}_3$ and $\mathrm{X}_4$ respectively. In a cell with extrinsic noise value $\vec{u}^i = (u^i_1,u^i_2)^T$, the production rate of $\mathrm{X}_1$ is $u^i_1 k_1$ and the production rate of $\mathrm{X}_2$ is $u^2_i k_3$. To model the constitutive reporters we define the augmented CRN $(\mathcal{R}_{aug}, \vec{g}_{aug}(\vec{u}), U, \vec{\alpha},\gamma)$ in Figure \ref{fig:ext_noise_antithetic}(b) which includes the reporter species $\mathrm{Y}_1$ and $\mathrm{Y}_2$. Therefore, we can use Theorem \ref{thm:ext_noise_cert_rre_only}. Considering $U \supseteq \left\{[0,1], [1,0], [1,1], [2,1], [2,2], [1,2]\right\}$ we find that for mixture component $i$ the reaction rate equations defined in \eqref{eq:RRE_def} are
\begin{align}
	0 & = \vec{f}(\vec{x}_i;\vec{k}),\\
	0 & = 
	\begin{bmatrix}
		u^i_{1}k_1 - k_2x_{1i} - k_5x_{1i}x_{2i} + k_6x_{2i}\\
		u^i_{2}k_3 - k_4x_{2i} - k_5x_{1i}x_{2i}
	\end{bmatrix}.
\end{align}
Where we use the notation $\vec{x}_i = [x_{1i}, x_{2i}]^T$. Forming $\bar{A}^{rre}(\vec{x}_1,\dots,\vec{x}_l,\vec{u}^1,\dots,\vec{u}^l)$ we find that \eqref{eq:extrinsic_noise_A_bar} is given by
\begin{equation}\label{eq:exFB_A_bar}
	0 = \bar{A}^{rre}(\vec{x}_1,\dots,\vec{x}_l,\vec{u}^1,\dots,\vec{u}^l)\vec{k}
	=
	\begin{bmatrix}
		1 & -x_{11} & 0 &    0 & -x_{11}x_{21} & x_{21}\\
		0 &    0 & 0 & -x_{21} & -x_{11}x_{21} &   0\\
		0 & -x_{12} & 0 &    0 & -x_{12}x_{22} & x_{22}\\
		0 &    0 & 1 & -x_{22} & -x_{12}x_{22} &   0\\
		1 & -x_{13} & 0 &    0 & -x_{13}x_{23} & x_{23}\\
		0 &    0 & 1 & -x_{23} & -x_{13}x_{23} &   0\\
		1 & -x_{14} & 0 &    0 & -x_{14}x_{24} & x_{24}\\
		0 &    0 & 2 & -x_{24} & -x_{14}x_{24} &   0\\
		2 & -x_{15} & 0 &    0 & -x_{15}x_{25} & x_{25}\\
		0 &    0 & 1 & -x_{25} & -x_{15}x_{25} &   0\\
		2 & -x_{16} & 0 &    0 & -x_{16}x_{26} & x_{26}\\
		0 &    0 & 2 & -x_{26} & -x_{16}x_{26} &   0
	\end{bmatrix}\vec{k}.
\end{equation}
The reduced Gr\"{o}bner basis of \eqref{eq:ideal_ext_constit} with $\bar{A}^{rre}(\vec{x}_1,\dots,\vec{x}_l,\vec{u}^1,\dots,\vec{u}^l)$ given by \eqref{eq:exFB_A_bar} is $\{1\}$, and hence, by Theorem \ref{thm:ext_noise_cert_rre_only}, $(\mathcal{R}_{aug}, \vec{g}_{aug}(\vec{u}), U,\vec{\alpha},1)$ is stationary globally identifiable over $\mathbb{R}^6_{>0}$.%

In this way the techniques of this paper help guide experimental design, since as shown in this example one can estimate all of the rate constants in this CRN from the stationary population distribution by placing the genes for $\mathrm{X}_1$ and $\mathrm{X}_2$ on separate plasmids, but not if the genes were e.g. genomically integrated in a single copy, or otherwise placed into the population of cells without copy number variation.

\begin{figure}[!h]
\centering
\subfloat[$(\mathcal{R},\vec{g}(\vec{u}),U)$]{%
       \centering
	\schemestart
		\subscheme{$\emptyset$}
    		\arrow(z--x1){<=>[$k_2$][$u^i_{1}k_1$]}[150]
		\subscheme{$\mathrm{X}_1$}
		\arrow(@z--x2){<=>[$u^i_{2}k_3$][$k_4$]}[90]
		\subscheme{$\mathrm{X}_2$}
		\arrow(@z--x12){<-[$k_5$]}[30]
		\subscheme{$\mathrm{X}_1 + \mathrm{X}_2$}
        \arrow(@x2--@x12){->[$k_6$]}
	\schemestop}
	\hspace{24pt}
  \subfloat[$(\mathcal{R}_{aug},\vec{g}_{aug}(\vec{u}),U)$]{%
        \centering
	\schemestart
		\subscheme{$\emptyset$}
    		\arrow(z--x1){<=>[$k_2$][$u^i_{1}k_1$]}[150]
		\subscheme{$\mathrm{X}_1$}
		\arrow(@z--x2){<=>[$u^i_{2}k_3$][$k_4$]}[90]
		\subscheme{$\mathrm{X}_2$}
		\arrow(@z--x12){<-[$k_5$]}[30]
		\subscheme{$\mathrm{X}_1 + \mathrm{X}_2$}
        \arrow(@x2--@x12){->[$k_6$]}
        \arrow(@z--z2){0}[0,2]
        \subscheme{$\emptyset$}
        \arrow(@z2--y1){<=>[$u^i_{1}\alpha$][$\gamma$]}[90]
        \subscheme{$\mathrm{Z}_1$}
        \arrow(@z2--z3){0}[0]
        \subscheme{$\emptyset$}
        \arrow(@z3--y2){<=>[$u^i_{2}\alpha$][$\gamma$]}[90]
        \subscheme{$\mathrm{Z}_2$}
        \schemestop}
    \\
	\caption{ The CRN with extrinsic noise $(\mathcal{R},\vec{g}(\vec{u}),\rho(\vec{u}),U)$ introduced in Example \ref{ex:feedback}. (a) Shows $(\mathcal{R},\vec{g}(\vec{u}),U)$ and (b) shows $(\mathcal{R}_{aug},\vec{g}_{aug}(\vec{u}),\rho(\vec{u}),U,\vec{\alpha},\gamma)$, the version augmented with constitutive reporters. Augmented CRN $(\mathcal{R}_{aug},\vec{g}_{aug}(\vec{u}),\rho(\vec{u}),U,\vec{\alpha},1)$ is stationary globally identifiable over $\mathbb{R}_{>0}^6$ if $U \supseteq \left\{[0,1], [1,0], [1,1], [2,1], [2,2], [1,2]\right\}$ and there is a constitutive promoter for $u_1$ and $u_2$.}\label{fig:ext_noise_antithetic}
\end{figure}
\end{example}

In this section we have studied the problem of checking if a CRN that is not necessarily stationary globally identifiable becomes identifiable when extrinsic noise is added. We now consider the converse problem, can the addition of extrinsic noise make an identifiable CRN become non identifiable? Here we give the following corollary, which formalizes the intuition that if a chemical reaction network without extrinsic noise is stationary globally identifiable, then adding extrinsic noise preserves identifiability as long as Condition \ref{cond:sigma_unique_cond} is met.
\begin{theorem}\label{coro:id_id_ext}
	Consider an augmented CRN with extrinsic noise $(\mathcal{R},\vec{g}(\vec{u}),\rho(\vec{u}),U, \vec{\alpha}, 1)$. Assume that $\forall \vec{u}\in U, \vec{g}(\vec{u}) > 0$. If the corresponding CRN without extrinsic noise $\mathcal{R}$ is stationary globally identifiable over $\mathbb{R}_{>0}^r$, then $(\mathcal{R},\vec{g}(\vec{u}),\rho(\vec{u}),U, \vec{\alpha}, 1)$ is stationary globally identifiable over $\mathbb{R}_{>0}^r$.
\end{theorem}
\begin{proof}
	Consider an arbitrary $\vec{x}_1,P_1$ that satisfies $0 = A(\vec{x}_1,P_1)\vec{g}(\vec{u}^1)\odot\vec{k}$ for some $\vec{k}\in \mathbb{R}_{>0}^r$. Letting $\vec{k}' = \vec{g}(\vec{u}^1)\odot\vec{k}$ we have that $0 = A(\vec{x}_1,P_1)\vec{k}'$ and $\vec{k}' \in \mathbb{R}_{>0}^r$. Therefore $\rank A(\vec{x}_1,P_1) = r-1$ by our assumption that $\mathcal{R}$ is stationary globally identifiable over $\mathbb{R}_{>0}^r$. Since $\rank A(\vec{x}_1,P_1) \diag \vec{g}(\vec{u}^1) = \rank A(\vec{x}_1,P_1)$, we have that $\bar{A}$ is rank $r-1$ for all $\vec{x}_1,\dots,\vec{x}_{l},P_1,\dots,P_{l}$ that satisfy $\bar{A}(\vec{x}_1,\dots,\vec{x}_{l},P_1,\dots,P_{l}, \vec{u}^1, \dots, \vec{u}^l)\vec{k}$ for some $\vec{k}\in \mathbb{R}_{>0}^r$. Therefore, the reduced Gr\"{o}bner basis of \eqref{eq:ideal_ext_constit} is $\{1\}$ and so by Theorem \ref{thm:grob_check_constit}, $(\mathcal{R},\vec{g}(\vec{u}),U, \vec{\alpha}, 1)$ is stationary globally identifiable over $\mathbb{R}_{>0}^r$.
\end{proof}
\setcounter{example}{9}
\begin{example}[1-dimensional extrinsic noise]
	Returning to Example \ref{ex:1d_extrinsic}, we now ask if we can conclude that $(\mathcal{R}_1,\vec{g}(\vec{u}),\rho(\vec{u}),U, \alpha, 1)$ with $\alpha >0$ is stationary globally identifiable simply by exploiting our results in Example \ref{ex:one_dim}. If we consider $(\mathcal{R}_1,\vec{g}(\vec{u}),\rho(\vec{u}),U', \alpha, 1)$, where $U' = \{1,2\}$, we can apply Theorem \ref{coro:id_id_ext} to conclude that since $\mathcal{R}_1$ is identifiable, the augmented CRN with extrinsic noise $(\mathcal{R}_1,\vec{g}(\vec{u}),\rho(\vec{u}),U', \alpha, 1)$ is also stationary globally identifiable. We note that if we used $U = \{0,1,2\}$ instead of $U'$, the condition $\vec{g}(\vec{u}) > 0$ would not be satisfied and so we would not be able to apply Theorem \ref{coro:id_id_ext}.
\end{example}

	We conclude with section by noting that while in general it is unclear how to verify Condition \ref{cond:sigma_unique_cond} for a non-augmented CRN with extrinsic noise, for the case $n=1$ and $s=1$, it is sometimes possible, as in the following example.

\setcounter{example}{9}
\begin{example}[1-dimensional extrinsic noise]
Here we continue Example \ref{ex:1d_extrinsic} and certify global stationary identifiability of $(\mathcal{R}_1,\vec{g}(\vec{u}),\rho(\vec{u}),U)$. Theorem \ref{thm:grob_check_constit} requires us to have an augmented network. However, if we can verify Condition \ref{cond:sigma_unique_cond} directly we can check identifiability by considering ideal \eqref{eq:ideal_ext_constit} directly. Here we consider $u_1 = \vec{u}\in U \subset \mathbb{R}$, and so we can write \eqref{eq:RRE_def} as
	\begin{equation}\label{eq:ex_1d_extrinsic_rre}
		\dot{x}_1 = u_1 k_1 - k_2x_1 - k_3x_1^2.
	\end{equation}
	If $u_1=0$, then the equilibrium value of $x_1$ is 0. Furthermore, letting $x_1^*$ denote the equilibrium of \eqref{eq:ex_1d_extrinsic_rre} we have that $\frac{\partial x_1^*}{\partial u}  = \frac{k_1}{k_2+2k_3x_1^*} > 0$. Therefore, the means of each mixture component in $f_X(\vec{x};\vec{k})$ are ordered such that if $u^i_1 < u^j_1$ then $x_i < x_j$. It follows that Condition \ref{cond:sigma_unique_cond} is satisfied, since given any $f(\cdot)$ of the form \eqref{eq:extrinsic_noise_stationary_dist}, $\mathcal{C}(f(\cdot)) = \{(w_1,x_1,p_1),(w_2,x_2,p_2),(w_3,x_3,p_3)\}$, where $x_1<x_2<x_3$, the only possible $\vec{\sigma}\in\Sigma^f$ consistent with $(\mathcal{R}, \vec{g}(\vec{u}),\rho(\vec{u}),U)$ is given by $\vec{\sigma}(0) = (w_1,x_1,p_1)$, $\vec{\sigma}(1) = (w_2,x_2,p_2)$, and $\vec{\sigma}(2) = (w_3,x_3,p_3)$. From \eqref{eq:Gamma_def} we have that for any value of $u_1 = \vec{u} \in U$
	\begin{equation}
		\Gamma(x; \vec{k})\Gamma(x; \vec{k})^T = u_1 k_1 + k_2x_1 + k_3x_1^2,
	\end{equation}
	and so, letting $\vec{x}_i= x_{i}$, $P_i = p_{i}$, $\vec{u}^1 = u^1_1 = 0$, $\vec{u}^2 = u^2_1 = 1$, and $\vec{u}^3 = u^3_1 = 2$, \eqref{eq:extrinsic_noise_A_bar} is given by
	\begin{equation}
		\bar{A}(\vec{x}_1,\vec{x}_2,\vec{x}_3, P_1,P_2,P_3) = 
		\begin{bmatrix}
		0 & -x_{1} & - x_{1}^2 \\
		0 & x_{1} - 2p_{1} & x_1^2 - 4p_{1}x_{1}\\
		1 & -x_{2} & - x_{2}^2 \\
		1 & x_{2} - 2p_{2} & x_{12}^2 - 4p_{2}x_{2}\\
		2 & -x_{3} & - x_{3}^2 \\
		2 & x_{3} - 2p_{3} & x_{3}^2 - 4p_{3}x_{3}
	\end{bmatrix}.
	\end{equation}
	We have established Condition \ref{cond:sigma_unique_cond} for this example, and hence we can establish global stationary identifiability by computing the reduced Gr\"{o}bner basis of the ideal \eqref{eq:ideal_ext_constit}, since in the proof of Theorem \ref{thm:grob_check_constit} the only place the augmented species are considered is in the verification of Condition \ref{cond:sigma_unique_cond}.
\end{example}

\section{Conclusion}
In this work we studied the identifiability of LNA models of chemical reaction networks with intrinsic and extrinsic noise from stationary distributions. We gave algebraic characterizations of identifiability and model discriminability which can be used to algorithmically prove identifiability or model discriminability holds for a given model. Our tools are therefore well suited to be used by practicing synthetic biologists and systems biologists to establish identifiability prior to running costly experiments, as well as to provide confidence that fitted parameters and inferred models are accurate. We applied our methods to many examples of biological relevance, those of which do not have extrinsic noise are summarized in Table \ref{tab:examples}. Since our results for chemical reaction networks with extrinsic noise require Condition \ref{cond:sigma_unique_cond}, which is in general difficult to verify unless the extrinsic noise arises from copy number variation and constitutive reporters are included in the CRN, future work includes algorithmic methods for checking Condition \ref{cond:sigma_unique_cond}.

\begin{table}
	\stepcounter{figure}
	\centering
	\begin{tabular}{|m{1.5em}|c|c|}
		Ex. & CRN & $K$ \\
		\hline %
		\hline
		\refstepcounter{rownumber} \label{tab:one_species} 
		1 &
		\schemestart
		\subscheme{$\emptyset$}
    		\arrow(z--x1){<=>[$k_1$][$k_2$]}[0]
		\subscheme{$\mathrm{X}_1$}
		\arrow(@x1--x11){<-[$k_3$]}[0]
		\subscheme{$2\mathrm{X}_1$}
	\schemestop & $\mathbb{R}^3_{>0}$ \refstepcounter{rownumber}\\
	\hline %
	\refstepcounter{rownumber} \label{tab:complex_balanced}
	3 &
	\schemestart
		\subscheme{$\emptyset$}
    		\arrow(z--x1){->[$k_1$]}[135]
		\subscheme{$\mathrm{X}_1$}
		\arrow(@z--x2){<-[$k_3$]}[45]
		\subscheme{$\mathrm{X}_2$}
		\arrow(@x1--@x2){->[$k_2$]}
	\schemestop & $\mathbb{R}^3_{>0}$ \\
	\hline %
	\refstepcounter{rownumber}  \label{tab:seq_rate}
	4 &
	\schemestart
		\subscheme{$\emptyset$}
    		\arrow(z--x1){<=>[$k_2$][$k_1$]}[135]
		\subscheme{$\mathrm{X}_1$}
		\arrow(@z--x2){<=>[$k_3$][$k_4$]}[90]
		\subscheme{$\mathrm{X}_2$}
		\arrow(@z--x12){<-[$k_5$]}[45]
		\subscheme{$\mathrm{X}_1 + \mathrm{X}_2$}
	\schemestop & $\mathbb{R}^5_{>0}$ \\
	\hline %
	\refstepcounter{rownumber}  \label{tab:high_ord_enzym_deg}
	5 &
	\schemestart
		\subscheme{$\emptyset$}
    		\arrow(z--x1){<=>[$k_2$][$k_1$]}[135]
		\subscheme{$\mathrm{X}_1$}
		\arrow(@z--x2){<=>[$k_3$][$k_4$]}[90]
		\subscheme{$\mathrm{X}_2$}
		\arrow(@x2--x12){0}[0,0.7]
		\subscheme{$2\mathrm{X}_1 + \mathrm{X}_2$}
		\arrow(@x12--x11){->[$k_5$]}[-90]
		\subscheme{$2\mathrm{X}_1$}
	\schemestop & $\mathbb{R}^5_{>0}$ \\
	\hline %
	\refstepcounter{rownumber}  \label{tab:three_node_path}
	6 &
	\schemestart
		\subscheme{$\emptyset$}
    		\arrow(z--x1){<=>[$k_1$][$k_2$]}[180, 0.7]
		\subscheme{$\mathrm{X}_1$}
          \arrow(@z--x2){<=>[$k_3$][$k_4$]}[90]
        \subscheme{$\mathrm{X}_2$}
          \arrow(@z--x3){<=>[$k_5$][$k_6$]}[0, 0.7]
        \subscheme{$\mathrm{X}_3$}
		  \arrow(@x1--x12){->[$k_7$]}[90]
        \subscheme{$\mathrm{X}_1 + \mathrm{X}_2$}
          \arrow(@x2--x23){->[$k_8$]}[0, 0.7]
		\subscheme{$\mathrm{X}_2 + \mathrm{X}_3$}
	\schemestop & $\mathbb{R}^8_{>0}$ \\
	\hline %
	\refstepcounter{rownumber}  \label{tab:two_seq}
	7 &
	\schemestart
		\subscheme{$\emptyset$}
    		\arrow(z--x1){<=>[$k_2$][$k_1$]}[135]
		\subscheme{$\mathrm{X}_1$}
		\arrow(@z--x2){->[$k_3$]}[90]
		\subscheme{$\mathrm{X}_2$}
		\arrow(@z--x3){->[$k_4$]}[45]
		\subscheme{$\mathrm{X}_3$}
		\arrow(@z--x12){<-[$k_5$]}[180, 0.7]
        \subscheme{$\mathrm{X}_1 + \mathrm{X}_2$}
		\arrow(@z--x23){<-[$k_6$]}[0, 0.7]
		\subscheme{$\mathrm{X}_2 + \mathrm{X}_3$}
	\schemestop & $\mathbb{R}^6_{>0}$ \\
	\hline
		\refstepcounter{rownumber}  \label{tab:which_activation}
		8 &
	\schemestart
		\subscheme{$\emptyset$}
    		\arrow(z--x1){<=>[$k_2$][$k_1$]}[180]
		\subscheme{$\mathrm{X}_1$}
          \arrow(@z--x2){<=>[$k_3$][$k_4$]}[0]
        \subscheme{$\mathrm{X}_2$}
		  \arrow(@x1--x12){->[$k_6$]}[90]
        \subscheme{$\mathrm{X}_1 + \mathrm{X}_2$}
          \arrow(@x2--x23){->[$k_5$]}[90]
		\subscheme{$\mathrm{X}_1 + \mathrm{X}_2$}
	\schemestop  & 
	$\begin{aligned}
		\{ \vec{k} \in \mathbb{R}_{\geq 0}^6 | \vec{k}_{1:4}>0,\; k_5 > 0\;&\mbox{and}\; k_6 = 0\;\\\quad \mbox{or} \; k_5 = 0\;&\mbox{and}\; k_6 > 0\}
		\end{aligned}$ \\
	\hline %
	\refstepcounter{rownumber}  \label{tab:seq_vs_enzym}
	9 &
	\schemestart
		\subscheme{$\emptyset$}
    		\arrow(z--x1){<=>[$k_2$][$k_1$]}[150]
		\subscheme{$\mathrm{X}_1$}
		\arrow(@z--x2){<=>[$k_3$][$k_4$]}[90]
		\subscheme{$\mathrm{X}_2$}
		\arrow(@z--x12){<-[$k_5$]}[30]
		\subscheme{$\mathrm{X}_1 + \mathrm{X}_2$}
        \arrow(@x2--@x12){<-[$k_6$]}
	\schemestop &
	$\begin{aligned}
		\{ \vec{k} \in \mathbb{R}_{\geq 0}^6 | \vec{k}_{1:4}>0,\; k_5 > 0\;&\mbox{and}\; k_6 = 0\;\\ \mbox{or} \; k_5 = 0\;&\mbox{and}\; k_6 > 0\}
		\end{aligned}$ \\
	\hline %
	\hline 
	\hline
	\end{tabular}
	\addtocounter{figure}{-1}
	\caption{Chemical reaction networks and the associated set $K$ over which stationary parametric identifiability has been certified using the techniques of Section \ref{sec:main_result}.}\label{tab:examples}
\end{table}

\section*{ACKNOWLEDGMENTS}
The authors thank Eduardo D. Sontag for several helpful discussions. This work was supported in part by the U.S. National Science Foundation under Grant CMMI grant 1727189 and U.S. AFOSR MURI under grant FA9550-22-1-0316.

\clearpage

\bibliographystyle{siamplain}
\bibliography{bibliography.bib}

\end{document}